\documentclass[11pt]{article}

\usepackage{graphicx,geometry,verbatim,color}
\usepackage{amsmath}
\usepackage{amssymb}
\usepackage{float}

\floatstyle{ruled}
\newfloat{algorithm}{thp}{lop}[section]
\newenvironment{proof}{\noindent\textbf{Proof}}{\hfill\qed}
\newcommand{\qed}{\hfill$\Box$}
\newtheorem{lemma}{Lemma}
\newtheorem{theorem}{Theorem}
\newtheorem{definition}{Definition}

\begin{document}

\title{Maximum Metric Spanning Tree made Byzantine Tolerant}

\author{Swan Dubois\protect\footnote{UPMC Sorbonne Universit\'es \& INRIA, France, swan.dubois@lip6.fr} \and Toshimitsu Masuzawa\protect\footnote{Osaka University, Japan, masuzawa@ist.osaka-u.ac.jp} \and S\'{e}bastien Tixeuil\protect\footnote{UPMC Sorbonne Universit\'es \& Institut Universitaire de France, France, sebastien.tixeuil@lip6.fr}}
\date{}

\maketitle

\begin{abstract}
Self-stabilization is a versatile approach to fault-tolerance since it permits a distributed system to recover from any transient fault that arbitrarily corrupts the contents of all memories in the system. Byzantine tolerance is an attractive feature of distributed systems that permits to cope with arbitrary malicious behaviors. This paper focus on systems that are both self-stabilizing and Byzantine tolerant.

We consider the well known problem of constructing a maximum metric tree in this context. Combining these two properties is known to induce many impossibility results. In this paper, we provide first two impossibility results about the construction of maximum metric tree in presence of transients and (permanent) Byzantine faults. Then, we provide a new self-stabilizing protocol that provides optimal containment of an arbitrary number of Byzantine faults.
\end{abstract}

\paragraph{Keywords}
Byzantine fault, Distributed protocol, Fault tolerance,
Stabilization, Spanning tree construction

\section{Introduction}

The advent of ubiquitous large-scale distributed systems advocates that tolerance to various kinds of faults and hazards must be included from the very early design of such systems. \emph{Self-stabilization}~\cite{D74j,D00b,T09bc} is a versatile technique that permits forward recovery from any kind of \emph{transient} faults, while \emph{Byzantine Fault-tolerance}~\cite{LSP82j} is traditionally used to mask the effect of a limited number of \emph{malicious} faults. Making distributed systems tolerant to both transient and malicious faults is appealing yet proved difficult~\cite{DW04j,DD05c,NA02c} as impossibility results are expected in many cases.

\paragraph{Related Works}
A promizing path towards multitolerance to both transient and Byzantine faults is \emph{Byzantine containment}. For \emph{local} tasks (\emph{i.e.} tasks whose correctness can be checked locally, such as vertex coloring, link coloring, or dining philosophers), the notion of \emph{strict stabilization} was proposed~\cite{NA02c,MT07j}. Strict stabilization guarantees that there exists a \emph{containment radius} outside which the effect of permanent faults is masked, provided that the problem specification makes it possible to break the causality chain that is caused by the faults. As many problems are not local, it turns out that it is impossible to provide strict stabilization for those. To circumvent impossibility results, the weaker notion of \emph{strong stabilization} was proposed~\cite{MT06cb,DMT11j}: here, correct nodes outside the containment radius may be perturbated by the actions of Byzantine node, but only a finite number of times.

Recently, the idea of generalizing strict and strong stabilization to an area that depends on the graph topology and the problem to be solved rather than an arbitrary fixed containment radius was proposed~\cite{DMT10ca,DMT10cd} and denoted by \emph{topology aware} strict (and strong) stabilization. When maximizable metric trees are considered, \cite{DMT10ca} proposed an optimal (with respect to impossibility results) protocol for topology-aware strict stabilization, and for the simpler case of breath-first-search metric trees, \cite{DMT10cd} presented a protocol that is optimal both with respect to strict and strong variants of topology-aware stabilization. The case of optimality for topology-aware strong stabilization in the general maximal metric case remains open.

\paragraph{Our Contribution} 

In this paper, we investigate the possibility of topology-aware strong stabilization for tasks that are global (\emph{i.e.} for with there exists a causality chain of size $r$, where $r$ depends on $n$ the size of the network), and focus on the maximum metric tree problem. Our contribution in this paper is threefold. First, we provide two impossibility results for self-stabilizing maximum metric tree construction in presence of Byzantine faults. In more details, we characterize a specific class of maximizable metrics (which includes breath-first-search and shortest path metrics) that prevents the existence of strong stabilizing solutions and we generalize an impossibilty result of \cite{DMT10cd} that provides a lower bound on the containmemt area for topology-aware strong stabilization (Section \ref{sec:impossibility}). Second, we provide a topology-aware strongly stabilizing protocol that matches this lower bound on the containment area (Section \ref{sec:protocol}). Finally, we provide a necessary and sufficient condition for the existence of a strongly stabilizing solution (Section \ref{sec:relationship}).

\section{Model, Definitions and Previous Results}

\subsection{State Model}

A \emph{distributed system} $S=(V,L)$ consists of a set $V=\{v_1,v_2,\ldots,v_n\}$ of processes and a set $L$ of bidirectional communication links (simply called links). A link is an unordered pair of distinct processes. A distributed system $S$ can be regarded as a graph whose vertex set is $V$ and whose link set is $L$, so we use graph terminology to describe a distributed system $S$. We use the following notations: $n=|V|$, $m=|L|$ and $d(u,v)$ denotes the distance between two processes $u$ and $v$ (\emph{i.e} the length of the shortest path between $u$ and $v$).

Processes $u$ and $v$ are called \emph{neighbors} if $(u,v)\in L$. The set of neighbors of a process $v$ is denoted by $N_v$. We do not assume existence of a unique identifier for each process. Instead we assume each process can distinguish its neighbors from each other by locally labeling them.

In this paper, we consider distributed systems of arbitrary topology. We assume that a single process is distinguished as a \emph{root}, and all the other processes are identical. We adopt the \emph{shared state model} as a communication model in this paper, where each process can directly read the states of its neighbors.

The variables that are maintained by processes denote process states. A process may take actions during the execution of the system. An action is simply a function that is executed in an atomic manner by the process. The action executed by each process is described by a finite set of guarded actions of the form $\langle$guard$\rangle\longrightarrow\langle$statement$\rangle$. Each guard of process $u$ is a boolean expression involving the variables of $u$ and its neighbors.

A global state of a distributed system is called a \emph{configuration} and is specified by a product of states of all processes. We define $C$ to be the set of all possible configurations of a distributed system $S$. For a process set $R \subseteq V$ and two configurations $\rho$ and $\rho'$, we denote $\rho \stackrel{R}{\mapsto} \rho'$ when $\rho$ changes to $\rho'$ by executing an action of each process in $R$ simultaneously. Notice that $\rho$ and $\rho'$ can be different only in the states of processes in $R$. For completeness of execution semantics, we should clarify the configuration resulting from simultaneous actions of neighboring processes. The action of a process depends only on its state at $\rho$ and the states of its neighbors at $\rho$, and the result of the action reflects on the state of the process at $\rho '$.

We say that a process is \emph{enabled} in a configuration $\rho$ if the guard of at least one of its actions is evaluated as true in $\rho$.

A \emph{schedule} of a distributed system is an infinite sequence of process sets. Let $Q=R^1, R^2, \ldots$  be a schedule, where $R^i \subseteq V$ holds for each $i\ (i \ge 1)$. An infinite sequence of configurations $e=\rho_0,\rho_1,\ldots$ is called an \emph{execution} from an initial configuration $\rho_0$ by a schedule $Q$, if $e$ satisfies $\rho_{i-1} \stackrel{R^i}{\mapsto} \rho_i$ for each $i\ (i \ge 1)$. Process actions are executed atomically, and we distinguish some properties on the scheduler (or daemon). A \emph{distributed daemon} schedules the actions of processes such that any subset of processes can simultaneously execute their actions. We say that the daemon is \emph{central} if it schedules action of only one process at any step. The set of all possible executions from $\rho_0\in C$ is denoted by $E_{\rho_0}$. The set of all possible executions is denoted by $E$, that is, $E=\bigcup_{\rho\in C}E_{\rho}$. We consider \emph{asynchronous} distributed systems but we add the following assumption on schedules: any schedule is strongly fair (that is, it is impossible for any process to be infinitely often enabled without executing its action in an execution) and $k$-bounded (that is, it is impossible for any process to execute more than $k$ actions between two consecutive action executions of any other process).

In this paper, we consider (permanent) \emph{Byzantine faults}: a Byzantine process (\emph{i.e.} a Byzantine-faulty process) can make arbitrary behavior independently from its actions. If $v$ is a Byzantine process, $v$ can repeatedly change its variables arbitrarily. For a given execution, the number of faulty processes is arbitrary but we assume that the root process is never faulty.

\subsection{Self-Stabilizing Protocols Resilient to Byzantine Faults}

Problems considered in this paper are so-called \emph{static problems}, \emph{i.e.} they require the system to find static solutions. For example, the spanning-tree construction problem is a static problem, while the mutual exclusion problem is not. Some static problems can be defined by a \emph{specification predicate} (shortly, specification), $spec(v)$, for each process $v$: a configuration is a desired one (with a solution) if every process satisfies $spec(v)$. A specification $spec(v)$ is a boolean expression on variables of $P_v~(\subseteq V)$ where $P_v$ is the set of processes whose variables appear in $spec(v)$. The variables appearing in the specification are called \emph{output variables} (shortly, \emph{O-variables}). In what follows, we consider a static problem defined by specification $spec(v)$.

A \emph{self-stabilizing protocol} (\cite{D74j}) is a protocol that eventually reaches a \emph{legitimate configuration}, where $spec(v)$ holds at every process $v$, regardless of the initial configuration. Once it reaches a legitimate configuration, every process never changes its O-variables and always satisfies $spec(v)$. From this definition, a self-stabilizing protocol is expected to tolerate any number and any type of transient faults since it can eventually recover from any configuration affected by the transient faults. However, the recovery from any configuration is guaranteed only when every process correctly executes its action from the configuration, \emph{i.e.}, we do not consider existence of permanently faulty processes.

When (permanent) Byzantine processes exist, Byzantine processes may not satisfy $spec(v)$. In addition, correct processes near the Byzantine processes can be influenced and may be unable to satisfy $spec(v)$. Nesterenko and Arora~\cite{NA02c} define a \emph{strictly stabilizing protocol} as a self-stabilizing protocol resilient to unbounded number of Byzantine processes.

Given an integer $c$, a \emph{$c$-correct process} is a process defined as follows.

\begin{definition}[$c$-correct process]
A process is $c$-correct if it is correct (\emph{i.e.} not Byzantine) and located at distance more than $c$ from any Byzantine process.
\end{definition}

\begin{definition}[$(c,f)$-containment]
\label{def:cfcontained}
A configuration $\rho$ is \emph{$(c,f)$-contained} for specification $spec$ if, given at most $f$ Byzantine processes, in any execution starting from $\rho$, every $c$-correct process $v$ always satisfies $spec(v)$ and never changes its O-variables.
\end{definition}

The parameter $c$ of Definition~\ref{def:cfcontained} refers to the \emph{containment radius} defined in \cite{NA02c}. The parameter $f$ refers explicitly to the number of Byzantine processes, while \cite{NA02c} dealt with unbounded number of Byzantine faults (that is $f\in\{0\ldots n\}$).

\begin{definition}[$(c,f)$-strict stabilization]
\label{def:cfstabilizing}
A protocol is \emph{$(c,f)$-strictly stabilizing} for specification $spec$ if, given at most $f$ Byzantine processes, any execution $e=\rho_0,\rho_1,\ldots$ contains a configuration $\rho_i$ that is $(c,f)$-contained for $spec$.
\end{definition}

An important limitation of the model of \cite{NA02c} is the notion of $r$-\emph{restrictive} specifications. Intuitively, a specification is $r$-restrictive if it prevents combinations of states that belong to two processes $u$ and $v$ that are at least $r$ hops away. An important consequence related to Byzantine tolerance is that the containment radius of protocols solving those specifications is at least $r$. For some (global) problems $r$ can not be bounded by a constant. In consequence, we can show that there exists no $(c,1)$-strictly stabilizing
protocol for such a problem for any (finite) integer $c$.

\paragraph{Strong stabilization} To circumvent such impossibility results, \cite{DMT11j} defines a weaker notion than the strict stabilization. Here, the requirement to the containment radius is relaxed, \emph{i.e.} there may exist processes outside the containment radius that invalidate the specification predicate, due to Byzantine actions. However, the impact of Byzantine triggered action is limited in times: the set of Byzantine processes may only impact processes outside the containment radius a bounded number of times, even if Byzantine processes execute an infinite number of actions.

In the following of this section, we recall the formal definition of strong stabilization adopted in \cite{DMT11j}. From the states of $c$-correct processes, \emph{$c$-legitimate configurations} and \emph{$c$-stable configurations} are defined as follows.

\begin{definition}[$c$-legitimate configuration]
A configuration $\rho$ is $c$-legitimate for \emph{spec} if every $c$-correct process $v$ satisfies $spec(v)$.
\end{definition}

\begin{definition}[$c$-stable configuration]
A configuration $\rho$ is $c$-stable if every $c$-correct process never changes the values of its O-variables as long as Byzantine processes make no action.
\end{definition}

Roughly speaking, the aim of self-stabilization is to guarantee that a distributed system eventually reaches a $c$-legitimate and $c$-stable configuration. However, a self-stabilizing system can be disturbed by Byzantine processes after reaching a $c$-legitimate and $c$-stable configuration. The \emph{$c$-disruption} represents the period where $c$-correct processes are disturbed by Byzantine processes and is defined as follows 

\begin{definition}[$c$-disruption]
A portion of execution $e=\rho_0,\rho_1,\ldots,\rho_t$ ($t>1$) is a $c$-disruption if and only if the following holds:
\begin{enumerate}
\item $e$ is finite,
\item $e$ contains at least one action of a $c$-correct process for changing the value of an O-variable,
\item $\rho_0$ is $c$-legitimate for \emph{spec} and $c$-stable, and
\item $\rho_t$ is the first configuration after $\rho_0$ such that $\rho_t$ is $c$-legitimate for \emph{spec} and $c$-stable.
\end{enumerate}
\end{definition}

Now we can define a self-stabilizing protocol such that Byzantine processes may only impact processes outside the containment radius a bounded number of times, even if Byzantine processes execute an infinite number of actions.

\begin{definition}[$(t,k,c,f)$-time contained configuration]
A configuration $\rho_0$ is $(t,k,c,f)$-time contained for \emph{spec} if given at most $f$ Byzantine processes, the following properties are satisfied:
\begin{enumerate}
\item $\rho_0$ is $c$-legitimate for \emph{spec} and $c$-stable,
\item every execution starting from $\rho_0$ contains a $c$-legitimate configuration for \emph{spec} after which the values of all the O-variables of $c$-correct processes remain unchanged (even when Byzantine processes make actions repeatedly and forever), 
\item every execution starting from $\rho_0$ contains at most $t$ $c$-disruptions, and 
\item every execution starting from $\rho_0$ contains at most $k$ actions of changing the values of O-variables for each $c$-correct process.
\end{enumerate}
\end{definition}

\begin{definition}[$(t,c,f)$-strongly stabilizing protocol]
A protocol $A$ is $(t,c,f)$-strongly stabilizing if and only if starting from any arbitrary configuration, every execution involving at most $f$ Byzantine processes contains a $(t,k,c,f)$-time contained configuration that is reached after at most $l$ rounds. Parameters $l$ and $k$ are respectively the $(t,c,f)$-stabilization time and the $(t,c,f)$-process-disruption times of $A$.
\end{definition}

Note that a $(t,k,c,f)$-time contained configuration is a $(c,f)$-contained configuration when $t=k=0$, and thus, $(t,k,c,f)$-time contained configuration is a generalization (relaxation) of a $(c,f)$-contained configuration. Thus, a strongly stabilizing protocol is weaker than a strictly stabilizing one (as processes outside the containment radius may take incorrect actions due to Byzantine influence). However, a strongly stabilizing protocol is stronger than a classical self-stabilizing one (that may never meet their specification in the presence of Byzantine processes).

The parameters $t$, $k$ and $c$ are introduced to quantify the strength of fault containment, we do not require each process to know the values of the parameters.

\paragraph{Topology-aware Byzantine resilience} We saw previously that there exist a number of impossibility results on strict stabilization due to the notion of $r$-restrictive specifications. To circumvent this impossibility result, we describe here another weaker notion than the strict stabilization: the \emph{topology-aware strict stabilization} (denoted by TA strict stabilization for short) introduced by \cite{DMT10ca}. Here, the requirement to the containment radius is relaxed, \emph{i.e.} the set of processes which may be disturbed by Byzantine ones is not reduced to the union of $c$-neighborhood of Byzantine processes (\emph{i.e.} the set of processes at distance at most $c$ from a Byzantine process) but can be defined depending on the graph topology and Byzantine processes location.

In the following, we give formal definition of this new kind of Byzantine containment. From now, $B$ denotes the set of Byzantine processes and $S_B$ (which is function of $B$) denotes a subset of $V$ (intuitively, this set gathers all processes which may be disturbed by Byzantine processes).

\begin{definition}[$S_{B}$-correct node]
A node is \emph{$S_{B}$-correct} if it is a correct node (\emph{i.e.} not Byzantine) which not belongs to $S_{B}$.
\end{definition}

\begin{definition}[$S_{B}$-legitimate configuration]
A configuration $\rho$ is \emph{$S_{B}$-legitimate} for $spec$ if every $S_{B}$-correct node $v$ is legitimate for $spec$ (\emph{i.e.} if $spec(v)$ holds).
\end{definition}

\begin{definition}[$(S_{B},f)$-topology-aware containment]
\label{def:SfTAcontained}
A configuration $\rho_{0}$ is \emph{$(S_{B},f)$-topology-aware contained} for specification $spec$ if, given at most $f$ Byzantine processes, in any execution $e=\rho_0,\rho_1,\ldots$, every configuration is $S_{B}$-legitimate and every $S_B$-correct process never changes its O-variables. 
\end{definition}

The parameter $S_{B}$ of Definition~\ref{def:SfTAcontained} refers to the \emph{containment area}. Any process which belongs to this set may be infinitely disturbed by Byzantine processes. The parameter $f$ refers explicitly to the number of Byzantine processes.

\begin{definition}[$(S_{B},f)$-topology-aware strict stabilization]
\label{def:SfTAStrictstabilizing}
A protocol is \emph{$(S_{B},f)$-topology-aware strictly stabilizing} for specification $spec$ if, given at most $f$ Byzantine processes, any execution $e=\rho_0,\rho_1,\ldots$ contains a configuration $\rho_i$ that is $(S_{B},f)$-topology-aware contained for $spec$.
\end{definition}

Note that, if $B$ denotes the set of Byzantine processes and $S_{B}=\left\{v\in V|\underset{b\in B}{min}\left(d(v,b)\right)\leq c\right\}$, then a $(S_{B},f)$-topology-aware strictly stabilizing protocol is a $(c,f)$-strictly stabilizing protocol. Then, the concept of topology-aware strict stabilization is a generalization of the strict stabilization. However, note that a TA strictly stabilizing protocol is stronger than a classical self-stabilizing protocol (that may never meet their specification in the presence of Byzantine processes). The parameter $S_{B}$ is introduced to quantify the strength of fault containment, we do not require each process to know the actual definition of the set.

Similarly to topology-aware strict stabilization, we can weaken the notion of strong stabilization using the notion of containment area. This idea was introduced by \cite{DMT10cd}. We recall in the following the formal definition of this concept.

\begin{definition}[$S_B$-stable configuration]
A configuration $\rho$ is $S_B$-stable if every $S_B$-correct process never changes the values of its O-variables as long as Byzantine processes make no action.
\end{definition}

\begin{definition}[$S_{B}$-TA-disruption]
A portion of execution $e=\rho_0,\rho_1,\ldots,\rho_t$ ($t>1$) is a $S_{B}$-TA-disruption if and only if the followings hold:
\begin{enumerate}
\item $e$ is finite,
\item $e$ contains at least one action of a $S_{B}$-correct process for changing the value of an O-variable,
\item $\rho_0$ is $S_{B}$-legitimate for $spec$ and $S_B$-stable, and
\item $\rho_t$ is the first configuration after $\rho_0$ such that $\rho_t$ is $S_{B}$-legitimate for $spec$ and $S_B$-stable.
\end{enumerate}
\end{definition}

\begin{definition}[$(t,k,S_{B},f)$-TA time contained configuration]
A configuration $\rho_0$ is $(t,k,S_{B},$ $f)$-TA time contained for \emph{spec} if given at most $f$ Byzantine processes, the following properties are satisfied:
\begin{enumerate}
\item $\rho_0$ is $S_{B}$-legitimate for \emph{spec} and $S_B$-stable,
\item every execution starting from $\rho_0$ contains a $S_B$-legitimate configuration for \emph{spec} after which the values of all the O-variables of $S_B$-correct processes remain unchanged (even when Byzantine processes make actions repeatedly and forever), 
\item every execution starting from $\rho_0$ contains at most $t$ $S_B$-TA-disruptions, and 
\item every execution starting from $\rho_0$ contains at most $k$ actions of changing the values of O-variables for each $S_B$-correct process.
\end{enumerate}
\end{definition}

\begin{definition}[$(t,S_{B},f)$-TA strongly stabilizing protocol]
A protocol $A$ is $(t,S_{B},f)$-TA\\ strongly stabilizing if and only if starting from any arbitrary configuration, every execution involving at most $f$ Byzantine processes contains a $(t,k,S_{B},f)$-TA-time contained configuration that is reached after at most $l$ rounds of each $S_{B}$-correct node. Parameters $l$ and $k$ are respectively the $(t,S_{B},f)$-stabilization time and the $(t,S_{B},f)$-process-disruption time of $A$.
\end{definition}

\subsection{Maximum Metric Tree Construction}

In this work, we deal with maximum (routing) metric trees as defined in \cite{GS03j}. Informally, the goal of a routing protocol is to construct a tree that simultaneously maximizes the metric values of all of the nodes with respect to some total ordering $\prec$. In the following, we recall all definitions and notations introduced in \cite{GS03j}. 

\begin{definition}[Routing metric]
A \emph{routing metric} (or just \emph{metric}) is a five-tuple $(M,W,met,mr,$ $\prec)$ where:
\begin{enumerate}
\item $M$ is a set of metric values,
\item $W$ is a set of edge weights,
\item $met$ is a metric function whose domain is $M\times W$ and whose range is $M$,
\item $mr$ is the maximum metric value in $M$ with respect to $\prec$ and is assigned to the root of the system,
\item $\prec$ is a less-than total order relation over $M$ that satisfies the following three conditions for arbitrary metric values $m$, $m'$, and $m''$ in $M$:
\begin{enumerate}
\item irreflexivity: $m\not\prec m$,
\item transitivity : if $m\prec m'$ and $m'\prec m''$ then $m\prec m''$,
\item totality: $m\prec m'$ or $m'\prec m$ or $m=m'$.
\end{enumerate}
\end{enumerate}
Any metric value $m\in M\setminus\{mr\}$ satisfies the \emph{utility condition} (that is, there exist $w_0,\ldots,w_{k-1}$ in $W$ and $m_0=mr,m_1,\ldots,m_{k-1},m_{k}=m$ in $M$ such that $\forall i\in\{1,\ldots,k\},m_i=met(m_{i-1},w_{i-1})$).
\end{definition}

For instance, we provide the definition of four classical metrics with this model: the shortest path metric ($\mathcal{SP}$), the flow metric ($\mathcal{F}$), and the reliability metric ($\mathcal{R}$). Note also that we can modelise the construction of a spanning tree with no particular constraints in this model using the metric $\mathcal{NC}$ described below and the construction of a BFS spanning tree using the shortest path metric ($\mathcal{SP}$) with $W_1=\{1\}$ (we denoted this metric by $\mathcal{BFS}$ in the following).

\[\begin{array}{rclrcl}
\mathcal{SP}&=&(M_1,W_1,met_1,mr_1,\prec_1)&\mathcal{F}&=&(M_2,W_2,met_2,mr_2,\prec_2)\\
\text{where}& & M_1=\mathbb{N}&\text{where}& & mr_2\in\mathbb{N}\\
&& W_1=\mathbb{N}&&& M_2=\{0,\ldots,mr_2\}\\
&& met_1(m,w)=m+w&&& W_2=\{0,\ldots,mr_2\}\\
&& mr_1=0&&& met_2(m,w)=min\{m,w\}\\
&& \prec_1 \text{ is the classical }>\text{ relation}&&& \prec_2 \text{ is the classical }<\text{ relation}
\end{array}\]
\[\begin{array}{rclrcl}
\mathcal{R}&=&(M_3,W_3,met_3,mr_3,\prec_3) & \mathcal{NC}&=&(M_4,W_4,met_4,mr_4,\prec_4)\\
\text{where}& & M_3=[0,1] & \text{where}& & M_4=\{0\} \\
&& W_3=[0,1] &&& W_4=\{0\}\\
&& met_3(m,w)=m*w &&& met_4(m,w)=0\\
&& mr_3=1 &&& mr_4=0\\
&& \prec_3 \text{ is the classical }<\text{ relation} &&& \prec_4 \text{ is the classical }<\text{ relation}
\end{array}\]

\begin{definition}[Assigned metric]
An \emph{assigned metric} over a system $S$ is a six-tuple $(M,W,met,$ $mr,\prec,wf)$ where $(M,W,met,mr,\prec)$ is a metric and $wf$ is a function that assigns to each edge of $S$ a weight in $W$.
\end{definition}

Let a rooted path (from $v$) be a simple path from a process $v$ to the root $r$. The next set of definitions are with respect to an assigned metric $(M,W,met,mr,\prec,wf)$ over a given system $S$.

\begin{definition}[Metric of a rooted path]
The \emph{metric of a rooted path} in $S$ is the prefix sum of $met$ over the edge weights in the path and $mr$.
\end{definition}

For example, if a rooted path $p$ in $S$ is $v_k,\ldots,v_0$ with $v_0=r$, then the metric of $p$ is $m_k=met(m_{k-1},wf(\{v_k,v_{k-1}\}))$ with $\forall i\in\{1,\ldots,k-1\},m_i=met(m_{i-1},wf(\{v_i,v_{i-1}\})$ and $m_0=mr$.

\begin{definition}[Maximum metric path]
A rooted path $p$ from $v$ in $S$ is called a \emph{maximum metric path} with respect to an assigned metric if and only if for every other rooted path $q$ from $v$ in $S$, the metric of $p$ is greater than or equal to the metric of $q$ with respect to the total order $\prec$. 
\end{definition}
 
\begin{definition}[Maximum metric of a node]
The \emph{maximum metric of a node} $v\neq r$ (or simply \emph{metric value} of $v$) in $S$ is defined by the metric of a maximum metric path from $v$. The maximum metric of $r$ is $mr$. 
\end{definition}

\begin{definition}[Maximum metric tree]
A spanning tree $T$ of $S$ is a \emph{maximum metric tree} with respect to an assigned metric over $S$ if and only if every rooted path in $T$ is a maximum metric path in $S$ with respect to the assigned metric.
\end{definition}

The goal of the work of \cite{GS03j} is the study of metrics that always allow the construction of a maximum metric tree. More formally, the definition follows.

\begin{definition}[Maximizable metric]
A metric is \emph{maximizable} if and only if for any assignment of this metric over any system $S$, there is a maximum metric tree for $S$ with respect to the assigned metric.
\end{definition}

Given a maximizable metric $\mathcal{M}=(M,W,mr,met,\prec)$, the aim of this work is to study the construction of a maximum metric tree with respect to $\mathcal{M}$ which spans the system in a self-stabilizing way in a system subject to permanent Byzantine faults (but we must assume that the root process is never a Byzantine one). It is obvious that these Byzantine processes may disturb some correct processes. It is why we relax the problem in the following way: we want to construct a maximum metric forest with respect to $\mathcal{M}$. The root of any tree of this forest must be either the real root or a Byzantine process. 

Each process $v$ has three O-variables: a pointer to its parent in its tree ($prnt_v\in N_v\cup\{\bot\}$), a level which stores its current metric value ($level_v\in M$) and an integer which stores a distance ($dist_v\in\mathbb{N}$). Obviously, Byzantine process may disturb (at least) their neighbors. We use the following specification of the problem.

We introduce new notations as follows. Given an assigned metric $(M,W,met,mr,\prec,wf)$ over the system $S$ and two processes $u$ and $v$, we denote by $\mu(u,v)$ the maximum metric of node $u$ when $v$ plays the role of the root of the system. If $u$ and $v$ are neighbors, we denote by $w_{u,v}$ the weight of the edge $\{u,v\}$ (that is, the value of $wf(\{u,v\})$).

\begin{definition}[$\mathcal{M}$-path]
Given an assigned metric $\mathcal{M}=(M,W,mr,met,\prec,wf)$ over a system $S$, a path $(v_0,\ldots,v_k)$ ($k\geq 1$) of $S$ is a \emph{$\mathcal{M}$-path} if and only if:
\begin{enumerate}
\item $prnt_{v_0}=\bot$, $level_{v_0}=mr$, $dist_{v_0}=0$, and $v_0\in B\cup\{r\}$,
\item $\forall i\in\{1,\ldots,k\}, prnt_{v_i}=v_{i-1}$ and $level_{v_i}=met(level_{v_{i-1}},w_{v_i,v_{i-1}})$,
\item $\forall i\in\{1,\ldots,k\}, met(level_{v_{i-1}},w_{v_i,v_{i-1}})=\underset{u\in N_v}{max_\prec}\{met(level_{u},w_{v_i,u})\}$,
\item $\forall i\in\{1,\ldots,k\}, dist_{v_i}=legal\_dist_{v_{i-1}}$ with $\forall u\in N_v, legal\_dist_u=\begin{cases} dist_{u}+1 \mbox{ if } level_{v}=level_{u}\\ 0 \mbox{ otherwise}\end{cases}$, and
\item $level_{v_{k}}=\mu(v_k,v_0)$.
\end{enumerate}
\end{definition}

We define the specification predicate $spec(v)$ of the maximum metric tree construction with respect to a maximizable metric $\mathcal{M}$ as follows.
\[spec(v) : \begin{cases}
 prnt_v = \bot \text{ and }  level_v = mr, \text{ and } dist_v=0 \text{ if } v \text{ is the root } r \\
 \text{there exists a }\mathcal{M}\text{-path } (v_0,\ldots,v_k) \text{ such that } v_k=v \text{ otherwise}
\end{cases}\]

\subsection{Previous results}

In this section, we summarize known results about maximum metric tree construction. The first interesting result about maximizable metrics is due to \cite{GS03j} that provides a fully characterization of maximizable metrics as follow.

\begin{definition}[Boundedness]
A metric $(M,W,met,mr,\prec)$ is \emph{bounded} if and only if: $\forall m \in M,\forall w\in W, met(m,w)\prec m \text{ or }met(m,w)=m$
\end{definition}

\begin{definition}[Monotonicity]
A metric $(M,W,met,mr,\prec)$ is \emph{monotonic} if and only if: $\forall (m,$ $m')\in M^2,\forall w\in W, m\prec m'\Rightarrow (met(m,w)\prec met(m',w)\text{ or }met(m,w)=met(m',w))$
\end{definition}

\begin{theorem}[Characterization of maximizable metrics \cite{GS03j}]
A metric is maximizable if and only if this metric is bounded and monotonic.
\end{theorem}

Secondly, \cite{GS99c} provides a self-stabilizing protocol to construct a maximum metric tree with respect to any maximizable metric. Now, we focus on self-stabilizating solutions resilient to Byzantine faults. Following discussion of Section 2, it is obvious that there exists no strictly stabilizing protocol for this problem. If we consider the weaker notion of topology-aware strict stabilization, \cite{DMT10ca} defines the best containment area as:

\[S_{B}=\left\{v\in V\setminus B\left|\mu(v,r)\preceq max_\prec\{\mu(v,b),b\in B\}\right.\right\}\setminus\{r\}\]

Intuitively, $S_B$ gathers correct processes that are closer (or at equal distance) from a Byzantine process than the root according to the metric. Moreover, \cite{DMT10ca} proves that the algorithm introduced for the maximum metric spanning tree construction in \cite{GS99c} performed this optimal containment area. More formally, \cite{DMT10ca} proves the following results.

\begin{theorem}[\cite{DMT10ca}]\label{th:impTAstrict}
Given a maximizable metric $\mathcal{M}=(M,W,mr,met,\prec)$, even under the central daemon, there exists no $(A_B,1)$-TA-strictly stabilizing protocol for maximum metric spanning tree construction with respect to $\mathcal{M}$ where $A_B\varsubsetneq S_B$.
\end{theorem}

\begin{theorem}[\cite{DMT10ca}]
Given a maximizable metric $\mathcal{M}=(M,W,mr,met,\prec)$, the protocol of \cite{GS99c} is a $(S_B,n-1)$-TA strictly stabilizing protocol for maximum metric spanning tree construction with respect to $\mathcal{M}$.
\end{theorem}

Some other works try to circumvent the impossibility result of strict stabilization using the concept ot strong stabilization but do not provide results for any maximizable metric. Indeed, \cite{DMT11j} proves the following result about spanning tree.

\begin{theorem}[\cite{DMT11j}]\label{th:possstrongNC}
There exists a $(t,0,n-1)$-strongly stabilizing protocol for maximum metric spanning tree construction with respect to $\mathcal{NC}$ (that is, for a spanning tree with no particular constraints) with a finite $t$.
\end{theorem}
 
On the other hand, regarding BFS spanning tree construction, \cite{DMT10cd} proved the following impossibility result.

\begin{theorem}[\cite{DMT10cd}]
Even under the central daemon, there exists no $(t,c,1)$-strongly stabilizing protocol for maximum metric spanning tree construction with respect to $\mathcal{BFS}$ where $t$ and $c$ are two finite integers.
\end{theorem}

Now, if we focus on topology-aware strong stabilization, \cite{DMT10cd} introduced the following containment area: $S_B^*=\{v\in V|\underset{b\in B}{min}(d(v,b))<d(r,v)\}$, and proved the following results.

\begin{theorem}[\cite{DMT10cd}]\label{th:impTAStrongBFS}
Even under the central daemon, there exists no $(t,A_B^*,1)$-TA strongly stabilizing protocol for maximum metric spanning tree construction with respect to $\mathcal{BFS}$ where $A_B^*\varsubsetneq S_B^*$ and $t$ is a finite integer.
\end{theorem}

\begin{theorem}[\cite{DMT10cd}]
The protocol of \cite{HC92j} is a $(t,S_B^*,n-1)$-TA strongly stabilizing protocol for maximum metric spanning tree construction with respect to $\mathcal{BFS}$ where $t$ is a finite integer.
\end{theorem}

The main motivation of this work is to fill the gap between results about TA strong and strong stabilization in the general case (that is, for any maximizable metric). Mainly, we define the best possible containment area for TA strong stabilization, we propose a protocol that provides this containment area and we characterize the set of metrics that allow strong stabilization.
 
\section{Impossibility Results}\label{sec:impossibility}

In this section, we provide our impossibility results about containment radius (respectively area) of any strongly stabilizing (respectively TA strongly stabilizing) protocol for the maximum metric tree construction.

\subsection{Strong Stabilization}

We introduce here some new definitions to characterize some important properties of maximizable metrics that are used in the following.

\begin{definition}[Strictly decreasing metric]
A metric $\mathcal{M}=(M,W,mr,met,\prec)$ is \emph{strictly decreasing} if, for any metric value $m\in M$, the following property holds: either $\forall w\in W,met(m,w)\prec m$ or $\forall w\in W,met(m,w)=m$.
\end{definition}

\begin{definition}[Fixed point]
A metric value $m$ is a \emph{fixed point} of a metric $\mathcal{M}=(M,W,mr,met,\prec)$ if $m\in M$ and if for any value $w\in W$, we have: $met(m,w)=m$.
\end{definition}

Then, we define a specific class of maximizable metrics and we prove that it is impossible to construct a maximum metric tree in a strongly-stabilizing way if we do not consider such a metric.

\begin{definition}[Strongly maximizable metric]
A maximizable metric $\mathcal{M}=(M,W,mr,met,\prec)$ is strongly maximizable if and only if $|M|=1$ or if the following properties holds: 
\begin{itemize}
\item $|M|\geq 2$,
\item $\mathcal{M}$ is strictly decreasing, and 
\item $\mathcal{M}$ has one and only one fixed point.
\end{itemize}
\end{definition}

Note that $\mathcal{NC}$ is a strongly maximizable metric (since $|M_4|=1$) whereas $\mathcal{BFS}$ or $\mathcal{SP}$ are not (since the first one has no fixed point, the second is not strictly decreasing). If we consider the metric $\mathcal{MET}$ defined below, we can show that $\mathcal{MET}$ is a strongly maximizable metric such that $|M|\geq 2$.

\[\begin{array}{rcl}
\mathcal{MET}&=&(M_5,W_5,met_5,mr_5,\prec_5)\\
\text{where}& & M_5=\{0,1,2,3\}\\
&& W_5=\{1\}\\
&& met_5(m,w)=max\{0,m-w\}\\
&& mr_5=3\\
&& \prec_5 \text{ is the classical }<\text{ relation}
\end{array}\]

Now, we can state our first impossibility result.

\begin{theorem}\label{th:necessarConditionStrong}
Given a maximizable metric $\mathcal{M}=(M,W,mr,met,\prec)$, even under the central daemon, there exists no $(t,c,1)$-strongly stabilizing protocol for maximum metric spanning tree construction with respect to $\mathcal{M}$ for any finite integer $t$ if:
\[\left\{\begin{array}{l}
\mathcal{M} \mbox{ is not a strongly maximizable metric, or}\\
c<|M|-2
\end{array}\right.\]
\end{theorem}

\begin{proof}
We prove this result by contradiction. We assume that $\mathcal{M}=(M,W,mr,met,\prec)$ is a maximizable metric such that there exist a finite integer $t$ and a protocol $\mathcal{P}$ that is a $(t,c,1)$-strongly stabilizing protocol for maximum metric spanning tree construction with respect to $\mathcal{M}$. We distinguish the following cases (note that they are exhaustive):
\begin{description}
\item[Case 1:] $\mathcal{M}$ is a strongly maximizing metric and $c<|M|-2$.

As $c\geq 0$, we know that $|M|\geq 2$ and by definition of a strongly stabilizing metric, $\mathcal{M}$ is strictly decreasing and has one and only one fixed point.

\begin{figure}[t]
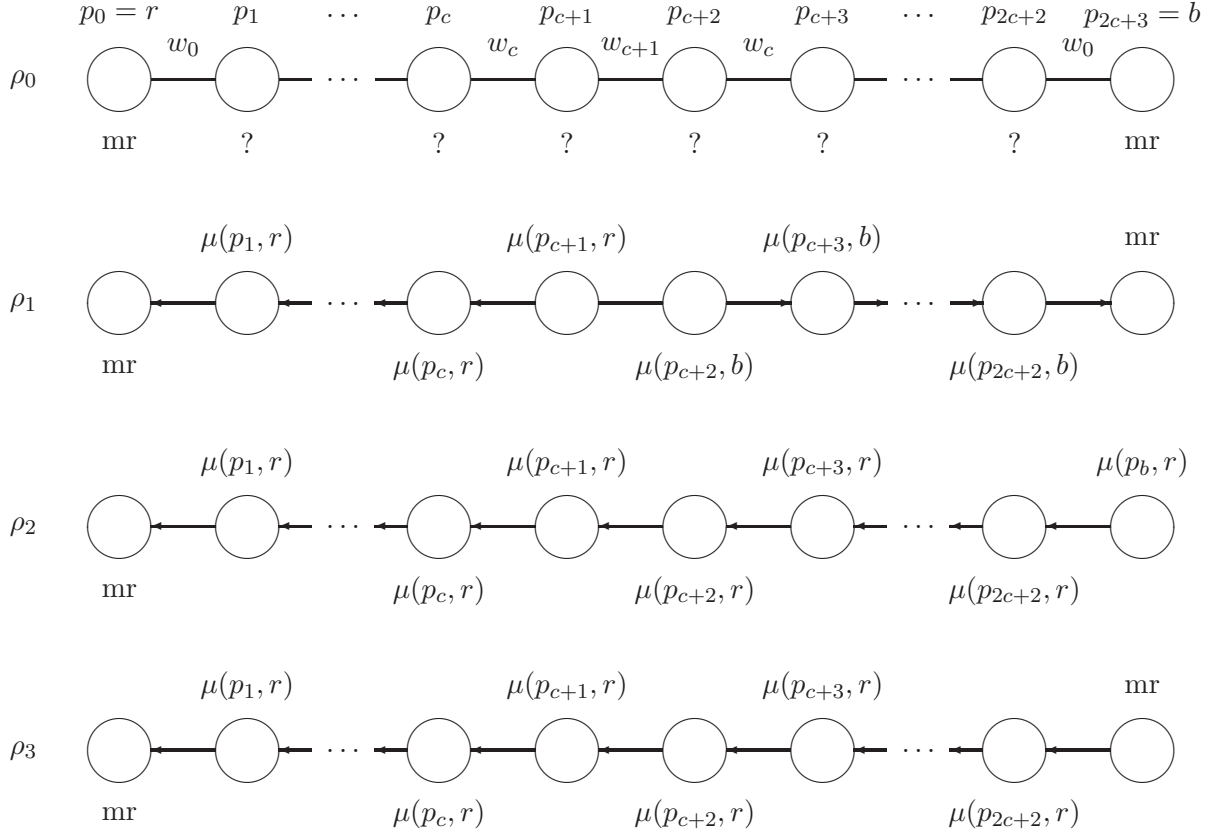

\noindent \begin{centering} \include{possStrongCase1}
  \par\end{centering}
 \caption{Configurations used in proof of Theorem \ref{th:necessarConditionStrong}, case 1.}
\label{fig:possStrongCase1}
\end{figure}

By assumption on $\mathcal{M}$, we know that there exist $c+3$ distinct metric values $m_0=mr,m_1,\ldots,$ $m_{c+2}$ in $M$ and $w_0,w_1,\ldots,w_{c+1}$ in $W$ such that: $\forall i\in\{1,\ldots,c+2\},m_i=met(m_{i-1},w_{i-1})\prec m_{i-1}$.

Let $S=(V,E,\mathcal{W})$ be the following weighted system $V=\{p_0=r,p_1,\ldots,p_{2c+2},p_{2c+3}=b\}$, $E=\{\{p_i,p_{i+1}\},i\in\{0,\ldots,2c+2\}\}$ and $\forall i\in\{0,c+1\},w_{p_i,p_{i+1}}=w_{p_{2c+3-i},p_{2c+2-i}}=w_i$. Note that the choice $w_{p_{c+1},p_{c+2}}=w_{c+1}$ ensures us the following property when $level_r=level_b=mr$: $\mu(p_{c+1},b)\prec\mu(p_{c+1},r)$ (and by symmetry, $\mu(p_{c+2},r)\prec\mu(p_{c+2},b)$). Process $p_0$ is the real root and process $b$ is a Byzantine one. Note that the construction of $\mathcal{W}$ ensures the following properties when $level_r=level_b=mr$: $\forall i\in\{1,\ldots,c+1\},\mu(p_i,r)=\mu(p_{2c+3-i},b)$, $\mu(p_i,b)\prec\mu(p_i,r)$ and $\mu(p_{2c+3-i},r)\prec\mu(p_{2c+3-i},b)$.

Assume that the initial configuration $\rho_0$ of $S$ satisfies: $prnt_r=prnt_b=\bot$, $level_r=level_b=mr$, and other variables of $b$ (in particular $dist$) are identical to those of $r$ (see Figure \ref{fig:possStrongCase1}, variables of other processes may be arbitrary). Assume now that $b$ takes exactly the same actions as $r$ (if any) immediately after $r$. Then, by symmetry of the execution and by convergence of $\mathcal{P}$ to $spec$, we can deduce that the system reaches in a finite time a configuration $\rho_1$ (see Figure \ref{fig:possStrongCase1}) in which: $\forall i\in\{1,\ldots,c+1\}, prnt_{p_i}=p_{i-1}$, $level_{p_i}=\mu(p_i,r)=m_i$, $dist_{p_i}=legal\_dist_{prnt_{p_i}}$ and $\forall i\in\{c+2,\ldots,2c+2\},prnt_{p_i}=p_{i+1}$, $level_{p_i}=\mu(p_{i},b)=m_{2c+3-i}$, and $dist_{p_i}=legal\_dist_{prnt_{p_i}}$ (because this configuration is the only one in which all correct process $v$ satisfies $spec(v)$ when $prnt_r=prnt_b=\bot$ and $level_r=level_b=mr$ by construction of $\mathcal{W}$). Note that $\rho_1$ is $c$-legitimate and $c$-stable.

Assume now that the Byzantine process acts as a correct process and executes correctly its algorithm. Then, by convergence of $\mathcal{P}$ in fault-free systems (remember that a strongly-stabilizing algorithm is a special case of self-stabilizing algorithm), we can deduce that the system reach in a finite time a configuration $\rho_2$ (see Figure \ref{fig:possStrongCase1}) in which: $\forall i\in\{1,\ldots,2c+3\},prnt_{p_i}=p_{i-1}$, $level_{p_i}=\mu(p_i,r)$, and $dist_{p_i}=legal\_dist_{prnt_{p_i}}$ (because this configuration is the only one in which all process $v$ satisfies $spec(v)$). Note that the portion of execution between $\rho_1$ and $\rho_2$ contains at least one $c$-perturbation ($p_{c+2}$ is a $c$-correct process and modifies at least once its O-variables) and that $\rho_2$ is $c$-legitimate and $c$-stable.

Assume now that the Byzantine process $b$ takes the following state: $prnt_{b}=\bot$ and $level_b=mr$. This step brings the system into configuration $\rho_3$ (see Figure \ref{fig:possStrongCase1}). From this configuration, we can repeat the execution we constructed from $\rho_0$. By the same token, we obtain an execution of $\mathcal{P}$ which contains $c$-legitimate and $c$-stable configurations (see $\rho_1$) and an infinite number of $c$-perturbation which contradicts the $(t,c,1)$-strong stabilization of $\mathcal{P}$.

\item[Case 2:] $\mathcal{M}$ is not strictly decreasing.

\begin{figure}[t]
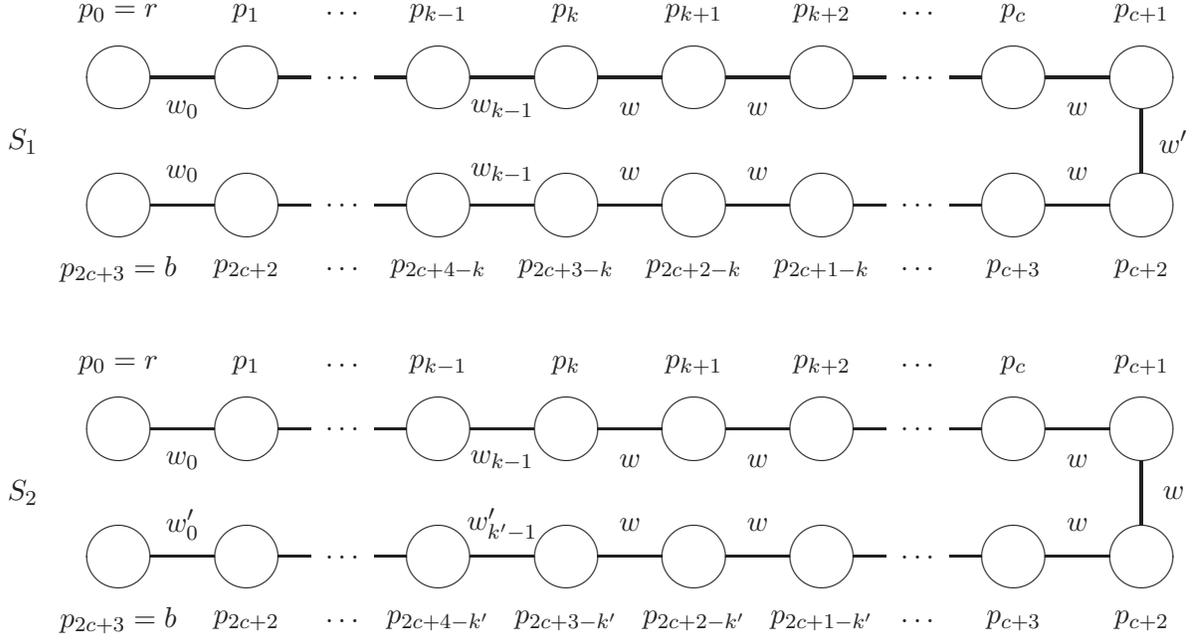

\noindent \begin{centering} \include{possStrongCase23}
  \par\end{centering}
 \caption{Configurations used in proof of Theorem \ref{th:necessarConditionStrong}, cases 2 and 3.}
\label{fig:possStrongCase23}
\end{figure}

By definition, we know that $\mathcal{M}$ is not a strongly maximizable metric. Hence, we have $|M|\geq 2$. Then, the definition of a strictly decreasing metric implies that there exists a metric value $m\in M$ such that: $\exists w\in W,$ $met(m,w)=m$ and $\exists w'\in W,m'=met(m,w')\prec m$ (and thus $m$ is not a fixed point of $\mathcal{M}$). By the utility condition on $M$, we know that there exists a sequence of metric values $m_0=mr,m_1,\ldots,m_l=m$ in $M$ and $w_0,w_1,\ldots,w_{l-1}$ in $W$ such that $\forall i\in\{1,\ldots,l\},m_i=met(m_{i-1},w_{i-1})$. Denote by $k$ the length of the shortest such sequence. Note that this implies that $\forall i\in\{1,\ldots,k\},m_i\prec m_{i-1}$ (otherwise we can remove $m_i$ from the sequence and this is contradictory with the construction of $k$). We distinguish the following cases:
\begin{description}
\item[Case 2.1:] $k\geq c+2$.\\
We can use the same token as case 1 above by using $w'$ instead of $w_{c+1}$ in the case where $k=c+2$ (since we know that $met(m,w')\prec m$).
\item[Case 2.2:] $k< c+2$.\\
Let $S_1=(V,E,\mathcal{W})$ be the following weighted system $V=\{p_0=r,p_1,\ldots,p_{2c+2},p_{2c+3}=b\}$, $E=\{\{p_i,p_{i+1}\},i\in\{0,\ldots,2c+2\}\}$, $\forall i\in\{0,\ldots,k-1\},w_{p_i,p_{i+1}}=w_{p_{2c+3-i},p_{2c+2-i}}=w_i$, $\forall i\in\{k,\ldots,c\},w_{p_i,p_{i+1}}=w_{p_{2c+3-i},p_{2c+2-i}}=w$ and $w_{p_{c+1},p_{c+2}}=w'$ (see Figure \ref{fig:possStrongCase23}). Note that this choice ensures us the following property when $level_r=level_b=mr$: $\mu(p_{c+1},b)\prec\mu(p_{c+1},r)$ (and by symmetry, $\mu(p_{c+2},r)\prec\mu(p_{c+2},b)$). Process $p_0$ is the real root and process $b$ is a Byzantine one. Note that the construction of $\mathcal{W}$ ensures the following properties when $level_r=level_b=mr$: $\forall i\in\{1,\ldots,c+1\},\mu(p_i,r)=\mu(p_{2c+3-i},b)$, $\mu(p_i,b)\prec\mu(p_i,r)$ and $\mu(p_{2c+3-i},r)\prec\mu(p_{2c+3-i},b)$.

This construction allows us to follow the same proof as in case 1 above.
\end{description}

\item[Case 3:] $\mathcal{M}$ has no or more than two fixed point, and is strictly decreasing.

If $\mathcal{M}$ has no fixed point and is strictly decreasing, then $|M|$ is not finite and then, we can apply the result of case 1 above since $c$ is a finite integer.

If $\mathcal{M}$ has two or more fixed points and is strictly decreasing, denote by $\Upsilon$ and $\Upsilon'$ two fixed points of $\mathcal{M}$. Without loss of generality, assume that $\Upsilon\prec\Upsilon'$. By the utility condition on $M$, we know that there exists sequences of metric values $m_0=mr,m_1,\ldots,m_l=\Upsilon$ and $m'_0=mr,m'_1,\ldots,m'_{l'}=\Upsilon'$ in $M$ and $w_0,w_1,\ldots,w_{l-1}$ and $w'_0,w'_1,\ldots,w'_{l'-1}$ in $W$ such that $\forall i\in\{1,\ldots,l\},m_i=met(m_{i-1},w_{i-1})$ and $\forall i\in\{1,\ldots,l'\},m'_i=met(m'_{i-1},w'_{i-1})$. Denote by $k$ and $k'$ the length of shortest such sequences. Note that this implies that $\forall i\in\{1,\ldots,k\},m_i\prec m_{i-1}$ and $\forall i\in\{1,\ldots,k'\},m'_i\prec m'_{i-1}$ (otherwise we can remove $m_i$ or $m'_i$ from the corresponding sequence). We distinguish the following cases:
\begin{description}
\item[Case 3.1:] $k>c+2$ or $k'>c+2$.\\
Without loss of generality, assume that $k>c+2$ (the second case is similar). We can use the same token as case 1 above.
\item[Case 3.2:] $k\leq c+2$ and $k'\leq c+2$.\\
Let $w$ be an arbitrary value of $W$. Let $S_2=(V,E,\mathcal{W})$ be the following weighted system $V=\{p_0=r,p_1,\ldots,p_{2c+2},p_{2c+3}=b\}$, $E=\{\{p_i,p_{i+1}\},i\in\{0,\ldots,2c+2\}\}$, $\forall i\in\{0,k-1\},w_{p_i,p_{i+1}}=w_i$, $\forall i\in\{0,k'-1\},w_{p_{2c+3-i},p_{2c+2-i}}=w'_i$ and $\forall i\in\{k,2c+2-k'\}, w_{p_{i},p_{i+1}}=w$ (see Figure \ref{fig:possStrongCase23}). Note that this choice ensures us the following property when $level_r=level_b=mr$: $\mu(p_{c+1},r)=\Upsilon\prec\Upsilon' =\mu(p_{c+1},b)$ and $\mu(p_{c+2},r)=\Upsilon\prec\Upsilon'=\mu(p_{c+2},b)$. Process $p_0$ is the real root and process $b$ is a Byzantine one. 

This construction allows us to follow a similar proof as in case 1 above (note that any process $u$ which satisfies $\mu(u,r)\prec\Upsilon'$ will be disturb infinitely often, in particular at least $p_{c+1}$ and $p_{c+2}$ which contradicts the $(t,c,1)$-strong stabilization of $\mathcal{P}$).
\end{description}
\end{description}
In any case, we show that there exists a system which contradicts the $(t,c,1)$-strong stabilization of $\mathcal{P}$ that ends the proof. 
\end{proof}

\subsection{Topology Aware Strong Stabilization}

First, we generalize the set $S_B^*$ previously defined for the $\mathcal{BFS}$ metric in \cite{DMT10cd} to any maximizable metric $\mathcal{M}=(M,W,mr,met,\prec)$.

\[S_{B}^*=\left\{v\in V\setminus B\left|\mu(v,r)\prec\underset{b\in B}{max_\prec}\{\mu(v,b)\}\right.\right\}\]

Intuitively, $S_B^*$ gathers the set of corrects processes that are strictly closer (according to $\mathcal{M}$) to a Byzantine process than the root. Figures from \ref{fig:ExSP} to \ref{fig:ExReliability} provide some examples of containment areas with respect to several maximizable metrics and compare it to $S_B$, the optimal containment area for TA strict stabilization. 

Note that we assume for the sake of clarity that $V\setminus S_B^*$ induces a connected subsystem. If it is not the case, then $S_B^*$ is extended to include all processes belonging to connected subsystems of $V\setminus S_B^*$ that not include $r$.

\begin{figure}[t]
\noindent \begin{centering} \ifx\JPicScale\undefined\def\JPicScale{0.75}\fi
\unitlength \JPicScale mm
\begin{picture}(165,95)(0,0)
\linethickness{0.3mm}
\put(50,85){\circle{10}}

\linethickness{0.3mm}
\put(30,65){\circle{10}}

\linethickness{0.3mm}
\put(70,65){\circle{10}}

\linethickness{0.3mm}
\put(30,45){\circle{10}}

\linethickness{0.3mm}
\put(70,45){\circle{10}}

\linethickness{0.3mm}
\put(50,25){\circle{10}}

\linethickness{0.3mm}
\multiput(30,70)(0.12,0.12){125}{\line(1,0){0.12}}
\linethickness{0.3mm}
\multiput(55,85)(0.12,-0.12){125}{\line(1,0){0.12}}
\linethickness{0.3mm}
\put(30,50){\line(0,1){10}}
\linethickness{0.3mm}
\multiput(30,40)(0.12,-0.12){125}{\line(1,0){0.12}}
\linethickness{0.3mm}
\multiput(55,25)(0.12,0.12){125}{\line(1,0){0.12}}
\linethickness{0.3mm}
\put(70,50){\line(0,1){10}}
\linethickness{0.3mm}
\put(35,65){\line(1,0){30}}
\linethickness{0.3mm}
\put(35,45){\line(1,0){30}}
\linethickness{0.3mm}
\multiput(35,65)(0.18,-0.12){167}{\line(1,0){0.18}}
\linethickness{0.3mm}
\put(140,85){\circle{10}}

\linethickness{0.3mm}
\put(120,65){\circle{10}}

\linethickness{0.3mm}
\put(160,65){\circle{10}}

\linethickness{0.3mm}
\put(120,45){\circle{10}}

\linethickness{0.3mm}
\put(160,45){\circle{10}}

\linethickness{0.3mm}
\put(140,25){\circle{10}}

\linethickness{0.3mm}
\multiput(120,70)(0.12,0.12){125}{\line(1,0){0.12}}
\linethickness{0.3mm}
\multiput(145,85)(0.12,-0.12){125}{\line(1,0){0.12}}
\linethickness{0.3mm}
\put(120,50){\line(0,1){10}}
\linethickness{0.3mm}
\multiput(120,40)(0.12,-0.12){125}{\line(1,0){0.12}}
\linethickness{0.3mm}
\multiput(145,25)(0.12,0.12){125}{\line(1,0){0.12}}
\linethickness{0.3mm}
\put(160,50){\line(0,1){10}}
\linethickness{0.3mm}
\put(125,65){\line(1,0){30}}
\linethickness{0.3mm}
\put(125,45){\line(1,0){30}}
\linethickness{0.3mm}
\multiput(125,65)(0.18,-0.12){167}{\line(1,0){0.18}}
\put(50,85){\makebox(0,0)[cc]{r}}

\put(140,85){\makebox(0,0)[cc]{r}}

\put(140,25){\makebox(0,0)[cc]{b}}

\put(50,25){\makebox(0,0)[cc]{b}}

\textcolor{blue}{
\put(115,25){\makebox(0,0)[cc]{$S_B^*$}}
}

\textcolor{red}{
\put(165,75){\makebox(0,0)[cc]{$S_B$}}
\put(20,25){\makebox(0,0)[cc]{$S_B=S_B^*$}}
}

\put(50,95){\makebox(0,0)[cc]{mr=0}}

\put(140,95){\makebox(0,0)[cc]{mr=0}}

\put(50,15){\makebox(0,0)[cc]{$level_b=0$}}

\put(140,15){\makebox(0,0)[cc]{$level_b=0$}}

\put(35,80){\makebox(0,0)[cc]{7}}

\put(65,80){\makebox(0,0)[cc]{6}}

\put(50,70){\makebox(0,0)[cc]{5}}

\put(75,55){\makebox(0,0)[cc]{4}}

\put(55,55){\makebox(0,0)[cc]{10}}

\put(50,40){\makebox(0,0)[cc]{8}}

\put(25,55){\makebox(0,0)[cc]{6}}

\put(65,30){\makebox(0,0)[cc]{32}}

\put(35,30){\makebox(0,0)[cc]{16}}

\textcolor{red}{
\linethickness{0.3mm}
\qbezier(25,30)(25.43,30.23)(30.48,32.17)
\qbezier(30.48,32.17)(35.53,34.11)(40,35)
\qbezier(40,35)(42.53,35.36)(44.98,35.21)
\qbezier(44.98,35.21)(47.43,35.05)(50,35)
\qbezier(50,35)(52.57,35.05)(55.02,35.21)
\qbezier(55.02,35.21)(57.47,35.36)(60,35)
\qbezier(60,35)(64.47,34.11)(69.52,32.17)
\qbezier(69.52,32.17)(74.57,30.23)(75,30)
}

\put(125,80){\makebox(0,0)[cc]{0}}

\put(155,80){\makebox(0,0)[cc]{0}}

\put(140,70){\makebox(0,0)[cc]{0}}

\put(165,55){\makebox(0,0)[cc]{0}}

\put(115,55){\makebox(0,0)[cc]{0}}

\put(155,30){\makebox(0,0)[cc]{0}}

\put(125,30){\makebox(0,0)[cc]{0}}

\put(140,40){\makebox(0,0)[cc]{0}}

\put(145,55){\makebox(0,0)[cc]{0}}

\textcolor{blue}{
\linethickness{0.3mm}
\qbezier(115,29.45)(115.43,29.68)(120.48,31.62)
\qbezier(120.48,31.62)(125.53,33.56)(130,34.45)
\qbezier(130,34.45)(132.53,34.81)(134.98,34.65)
\qbezier(134.98,34.65)(137.43,34.5)(140,34.45)
\qbezier(140,34.45)(142.57,34.5)(145.02,34.65)
\qbezier(145.02,34.65)(147.47,34.81)(150,34.45)
\qbezier(150,34.45)(154.47,33.56)(159.52,31.62)
\qbezier(159.52,31.62)(164.57,29.68)(165,29.45)
}

\textcolor{red}{
\linethickness{0.3mm}
\qbezier(115,80)(115.43,79.77)(120.48,77.83)
\qbezier(120.48,77.83)(125.53,75.89)(130,75)
\qbezier(130,75)(135.08,74.17)(140,74.17)
\qbezier(140,74.17)(144.92,74.17)(150,75)
\qbezier(150,75)(153.99,75.69)(157.63,77.04)
\qbezier(157.63,77.04)(161.26,78.4)(165,80)
}
\end{picture}
  \par\end{centering}
 \caption{Examples of containment areas for $\mathcal{SP}$.}
\label{fig:ExSP}
\end{figure}
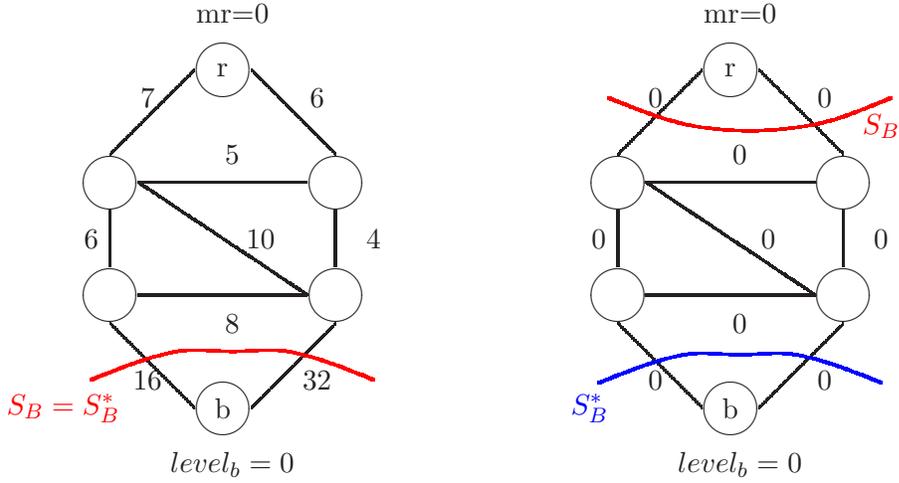

\begin{figure}[t]
\noindent \begin{centering} \ifx\JPicScale\undefined\def\JPicScale{0.75}\fi
\unitlength \JPicScale mm
\begin{picture}(173,95)(0,0)
\linethickness{0.3mm}
\put(50,85){\circle{10}}

\linethickness{0.3mm}
\put(30,65){\circle{10}}

\linethickness{0.3mm}
\put(70,65){\circle{10}}

\linethickness{0.3mm}
\put(30,45){\circle{10}}

\linethickness{0.3mm}
\put(70,45){\circle{10}}

\linethickness{0.3mm}
\put(50,25){\circle{10}}

\linethickness{0.3mm}
\multiput(30,70)(0.12,0.12){125}{\line(1,0){0.12}}
\linethickness{0.3mm}
\multiput(55,85)(0.12,-0.12){125}{\line(1,0){0.12}}
\linethickness{0.3mm}
\put(30,50){\line(0,1){10}}
\linethickness{0.3mm}
\multiput(30,40)(0.12,-0.12){125}{\line(1,0){0.12}}
\linethickness{0.3mm}
\multiput(55,25)(0.12,0.12){125}{\line(1,0){0.12}}
\linethickness{0.3mm}
\put(70,50){\line(0,1){10}}
\linethickness{0.3mm}
\put(35,65){\line(1,0){30}}
\linethickness{0.3mm}
\put(35,45){\line(1,0){30}}
\linethickness{0.3mm}
\multiput(35,65)(0.18,-0.12){167}{\line(1,0){0.18}}
\linethickness{0.3mm}
\put(140,85){\circle{10}}

\linethickness{0.3mm}
\put(120,65){\circle{10}}

\linethickness{0.3mm}
\put(160,65){\circle{10}}

\linethickness{0.3mm}
\put(120,45){\circle{10}}

\linethickness{0.3mm}
\put(160,45){\circle{10}}

\linethickness{0.3mm}
\put(140,25){\circle{10}}

\linethickness{0.3mm}
\multiput(120,70)(0.12,0.12){125}{\line(1,0){0.12}}
\linethickness{0.3mm}
\multiput(145,85)(0.12,-0.12){125}{\line(1,0){0.12}}
\linethickness{0.3mm}
\put(120,50){\line(0,1){10}}
\linethickness{0.3mm}
\multiput(120,40)(0.12,-0.12){125}{\line(1,0){0.12}}
\linethickness{0.3mm}
\multiput(145,25)(0.12,0.12){125}{\line(1,0){0.12}}
\linethickness{0.3mm}
\put(160,50){\line(0,1){10}}
\linethickness{0.3mm}
\put(125,65){\line(1,0){30}}
\linethickness{0.3mm}
\put(125,45){\line(1,0){30}}
\linethickness{0.3mm}
\multiput(125,65)(0.18,-0.12){167}{\line(1,0){0.18}}
\put(50,85){\makebox(0,0)[cc]{r}}

\put(140,85){\makebox(0,0)[cc]{r}}

\put(140,25){\makebox(0,0)[cc]{b}}

\put(50,25){\makebox(0,0)[cc]{b}}

\put(50,95){\makebox(0,0)[cc]{mr=10}}

\put(140,95){\makebox(0,0)[cc]{mr=10}}

\put(35,80){\makebox(0,0)[cc]{7}}

\put(65,80){\makebox(0,0)[cc]{6}}

\put(50,70){\makebox(0,0)[cc]{5}}

\put(75,52){\makebox(0,0)[cc]{4}}

\put(55,55){\makebox(0,0)[cc]{10}}

\put(25,55){\makebox(0,0)[cc]{6}}

\put(50,40){\makebox(0,0)[cc]{8}}

\put(65,30){\makebox(0,0)[cc]{32}}

\put(30,30){\makebox(0,0)[cc]{16}}

\put(140,15){\makebox(0,0)[cc]{$level_b=10$}}

\put(50,15){\makebox(0,0)[cc]{$level_b=10$}}

\put(120,30){\makebox(0,0)[cc]{11}}

\put(155,30){\makebox(0,0)[cc]{12}}

\put(155,80){\makebox(0,0)[cc]{10}}

\put(120,80){\makebox(0,0)[cc]{7}}

\put(165,55){\makebox(0,0)[cc]{13}}

\put(140,70){\makebox(0,0)[cc]{6}}

\put(145,55){\makebox(0,0)[cc]{5}}

\put(115,55){\makebox(0,0)[cc]{3}}

\put(140,40){\makebox(0,0)[cc]{1}}

\textcolor{blue}{
\linethickness{0.3mm}
\qbezier(105,48.29)(105.59,48.67)(112.37,51.07)
\qbezier(112.37,51.07)(119.14,53.48)(125,53.29)
\qbezier(125,53.29)(127.81,52.88)(130.18,51.35)
\qbezier(130.18,51.35)(132.54,49.83)(135,48.29)
\qbezier(135,48.29)(138.94,45.85)(142.39,43.06)
\qbezier(142.39,43.06)(145.84,40.27)(150,38.29)
\qbezier(150,38.29)(154.85,36.17)(159.79,35.03)
\qbezier(159.79,35.03)(164.73,33.9)(170,33.29)
\put(105,44.14){\makebox(0,0)[cc]{$S_B^*$}}
}

\textcolor{red}{
\linethickness{0.3mm}
\qbezier(105,53.86)(105.43,54.17)(110.46,56.36)
\qbezier(110.46,56.36)(115.5,58.54)(120,58.86)
\qbezier(120,58.86)(124.01,58.43)(127.58,55.89)
\qbezier(127.58,55.89)(131.14,53.36)(135,53.86)
\qbezier(135,53.86)(139.98,55.48)(142.9,60.37)
\qbezier(142.9,60.37)(145.83,65.27)(150,68.86)
\qbezier(150,68.86)(152.33,70.57)(154.77,71.87)
\qbezier(154.77,71.87)(157.21,73.18)(160,73.86)
\qbezier(160,73.86)(162.53,74.5)(165.07,74.64)
\qbezier(165.07,74.64)(167.61,74.79)(170,73.86)
\put(173,68.71){\makebox(0,0)[cc]{$S_B$}}
}

\textcolor{red}{
\linethickness{0.3mm}
\qbezier(20,75)(20.66,75.25)(27.63,75.98)
\qbezier(27.63,75.98)(34.6,76.72)(40,75)
\qbezier(40,75)(41.65,74.26)(42.8,72.89)
\qbezier(42.8,72.89)(43.96,71.52)(45,70)
\qbezier(45,70)(46.51,67.55)(47.23,64.72)
\qbezier(47.23,64.72)(47.94,61.9)(50,60)
\qbezier(50,60)(54.2,56.72)(59.43,55.56)
\qbezier(59.43,55.56)(64.65,54.39)(70,55)
\qbezier(70,55)(72.83,55.38)(75.26,56.81)
\qbezier(75.26,56.81)(77.68,58.25)(80,60)
\put(10,70){\makebox(0,0)[cc]{$S_B=S_B^*$}}
}
\end{picture}
  \par\end{centering}
 \caption{Examples of containment areas for $\mathcal{F}$.}
\label{fig:ExFlow}
\end{figure}

\begin{figure}[t]
\noindent \begin{centering} \ifx\JPicScale\undefined\def\JPicScale{0.75}\fi
\unitlength \JPicScale mm
\begin{picture}(175,95)(0,0)
\linethickness{0.3mm}
\put(50,85){\circle{10}}

\linethickness{0.3mm}
\put(30,65){\circle{10}}

\linethickness{0.3mm}
\put(70,65){\circle{10}}

\linethickness{0.3mm}
\put(30,45){\circle{10}}

\linethickness{0.3mm}
\put(70,45){\circle{10}}

\linethickness{0.3mm}
\put(50,25){\circle{10}}

\linethickness{0.3mm}
\multiput(30,70)(0.12,0.12){125}{\line(1,0){0.12}}
\linethickness{0.3mm}
\multiput(55,85)(0.12,-0.12){125}{\line(1,0){0.12}}
\linethickness{0.3mm}
\put(30,50){\line(0,1){10}}
\linethickness{0.3mm}
\multiput(30,40)(0.12,-0.12){125}{\line(1,0){0.12}}
\linethickness{0.3mm}
\multiput(55,25)(0.12,0.12){125}{\line(1,0){0.12}}
\linethickness{0.3mm}
\put(70,50){\line(0,1){10}}
\linethickness{0.3mm}
\put(35,65){\line(1,0){30}}
\linethickness{0.3mm}
\put(35,45){\line(1,0){30}}
\linethickness{0.3mm}
\multiput(35,65)(0.18,-0.12){167}{\line(1,0){0.18}}
\linethickness{0.3mm}
\put(140,85){\circle{10}}

\linethickness{0.3mm}
\put(120,65){\circle{10}}

\linethickness{0.3mm}
\put(160,65){\circle{10}}

\linethickness{0.3mm}
\put(120,45){\circle{10}}

\linethickness{0.3mm}
\put(160,45){\circle{10}}

\linethickness{0.3mm}
\put(140,25){\circle{10}}

\linethickness{0.3mm}
\multiput(120,70)(0.12,0.12){125}{\line(1,0){0.12}}
\linethickness{0.3mm}
\multiput(145,85)(0.12,-0.12){125}{\line(1,0){0.12}}
\linethickness{0.3mm}
\put(120,50){\line(0,1){10}}
\linethickness{0.3mm}
\multiput(120,40)(0.12,-0.12){125}{\line(1,0){0.12}}
\linethickness{0.3mm}
\multiput(145,25)(0.12,0.12){125}{\line(1,0){0.12}}
\linethickness{0.3mm}
\put(160,50){\line(0,1){10}}
\linethickness{0.3mm}
\put(125,65){\line(1,0){30}}
\linethickness{0.3mm}
\put(125,45){\line(1,0){30}}
\linethickness{0.3mm}
\multiput(125,65)(0.18,-0.12){167}{\line(1,0){0.18}}
\put(50,85){\makebox(0,0)[cc]{r}}

\put(140,85){\makebox(0,0)[cc]{r}}

\put(140,25){\makebox(0,0)[cc]{b}}

\put(50,25){\makebox(0,0)[cc]{b}}

\textcolor{blue}{
\put(105,40){\makebox(0,0)[cc]{$S_B^*$}}
}

\textcolor{red}{
\put(175,75){\makebox(0,0)[cc]{$S_B$}}
\put(15,40){\makebox(0,0)[cc]{$S_B=S_B^*$}}
}

\put(50,95){\makebox(0,0)[cc]{mr=1}}

\put(50,15){\makebox(0,0)[cc]{$level_b=1$}}

\put(140,95){\makebox(0,0)[cc]{mr=1}}

\put(140,15){\makebox(0,0)[cc]{$level_b=1$}}

\put(70,80){\makebox(0,0)[cc]{0,75}}

\put(30,80){\makebox(0,0)[cc]{0,75}}

\put(70,30){\makebox(0,0)[cc]{0,75}}

\put(30,30){\makebox(0,0)[cc]{0,75}}

\put(50,40){\makebox(0,0)[cc]{1}}

\put(50,70){\makebox(0,0)[cc]{1}}

\put(50,60){\makebox(0,0)[cc]{0,8}}

\put(25,55){\makebox(0,0)[cc]{0,4}}

\put(75,55){\makebox(0,0)[cc]{0,3}}

\textcolor{red}{
\linethickness{0.3mm}
\qbezier(15,45)(15.38,45.46)(20.28,49.33)
\qbezier(20.28,49.33)(25.17,53.2)(30,55)
\qbezier(30,55)(33.69,55.98)(37.41,55.56)
\qbezier(37.41,55.56)(41.13,55.14)(45,55)
\qbezier(45,55)(51.45,55.23)(57.65,55.93)
\qbezier(57.65,55.93)(63.85,56.63)(70,55)
\qbezier(70,55)(74.83,53.2)(79.72,49.33)
\qbezier(79.72,49.33)(84.62,45.46)(85,45)
}

\put(120,80){\makebox(0,0)[cc]{0,25}}

\put(140,70){\makebox(0,0)[cc]{0,25}}

\put(160,80){\makebox(0,0)[cc]{0,75}}

\put(165,55){\makebox(0,0)[cc]{1}}

\put(145,55){\makebox(0,0)[cc]{0,5}}

\put(140,40){\makebox(0,0)[cc]{1}}

\put(110,55){\makebox(0,0)[cc]{0,25}}

\put(160,30){\makebox(0,0)[cc]{0,75}}

\put(120,30){\makebox(0,0)[cc]{0,5}}

\textcolor{blue}{
\linethickness{0.3mm}
\qbezier(105,45)(105.38,45.46)(110.28,49.33)
\qbezier(110.28,49.33)(115.17,53.2)(120,55)
\qbezier(120,55)(129.9,57.99)(140,57.99)
\qbezier(140,57.99)(150.1,57.99)(160,55)
\qbezier(160,55)(164.83,53.2)(169.72,49.33)
\qbezier(169.72,49.33)(174.62,45.46)(175,45)
}

\textcolor{red}{
\linethickness{0.3mm}
\qbezier(105,50)(105.87,50.46)(115.97,54.35)
\qbezier(115.97,54.35)(126.07,58.23)(135,60)
\qbezier(135,60)(137.59,60.19)(140.19,59.54)
\qbezier(140.19,59.54)(142.8,58.88)(145,60)
\qbezier(145,60)(147.18,61.68)(147.69,64.73)
\qbezier(147.69,64.73)(148.21,67.78)(150,70)
\qbezier(150,70)(152.05,72.01)(154.61,73.19)
\qbezier(154.61,73.19)(157.17,74.36)(160,75)
\qbezier(160,75)(162.5,75.52)(164.96,75.3)
\qbezier(164.96,75.3)(167.42,75.07)(170,75)
}
\end{picture}
  \par\end{centering}
 \caption{Examples of containment areas for $\mathcal{R}$.}
\label{fig:ExReliability}
\end{figure}

Now, we can state our generalization of Theorem \ref{th:impTAStrongBFS}.

\begin{theorem}\label{th:impTAstrong}
Given a maximizable metric $\mathcal{M}=(M,W,mr,met,\prec)$, even under the central daemon, there exists no $(t,A_B^*,1)$-TA-strongly stabilizing protocol for maximum metric spanning tree construction with respect to $\mathcal{M}$ where $A_B^*\varsubsetneq S_B^*$ and $t$ is a given finite integer.
\end{theorem}

\begin{proof}
Let $\mathcal{M}=(M,W,mr,met,\prec)$ be a maximizable metric and $\mathcal{P}$ be a $(t,A_B^*,1)$-TA-strongly stabilizing protocol for maximum metric spanning tree construction protocol with respect to $\mathcal{M}$ where $A_B^*\varsubsetneq S_B^*$ and $t$ is a finite integer. We must distinguish the following cases:

\begin{description}
\item[Case 1:] $|M|=1$.\\
Denote by $m$ the metric value such that $M=\{m\}$. For any system and for any process $v$, we have $\mu(v,r)=\underset{b\in B}{min_\prec}\{\mu(v,b)\}=m$. Consequently, $S_B^*=\emptyset$ for any system. Then, it is absurd to have $A_B^*\varsubsetneq S_B^*$.
\item[Case 2:] $|M|\geq 2$.\\
By definition of a bounded metric, we can deduce that there exists $m\in M$ and $w\in W$ such that $m=met(mr,w)\prec mr$. Then, we must distinguish the following cases:
\begin{description}
\item[Case 2.1:] $m$ is a fixed point of $\mathcal{M}$.\\
Let $S$ be a system such that any edge incident to the root or a Byzantine process has a weight equals to $w$. Then, we can deduce that we have: $m=\underset{b\in B}{max_\prec}\{\mu(r,b)\}\prec \mu(r,r)=mr$ and for any correct process $v\neq r$, $\mu(v,r)=\underset{b\in B}{max_\prec}\{\mu(v,b)\}=m$. Hence, $S_B^*=\emptyset$ for any such system. Then, it is absurd to have $A_B^*\varsubsetneq S_B^*$.
\item[Case 2.2:] $m$ is not a fixed point of $\mathcal{M}$.\\
This implies that there exists $w'\in W$ such that: $met(m,w')\prec m$ (remember that $\mathcal{M}$ is bounded). Consider the following system: $V=\{r,u,u',v,v',b\}$, $E=\{\{r,u\},\{r,u'\},$ $\{u,v\},\{u',v'\},\{v,b\},\{v',b\}\}$, $w_{r,u}=w_{r,u'}=w_{v,b}=w_{v',b}=w$, and $w_{u,v}=w_{u',v'}=w'$ ($b$ is a Byzantine process). We can see that $S_B^*=\{v,v'\}$. Since $A_B^*\varsubsetneq S_B$, we have: $v\notin A_B^*$ or $v'\notin A_B^*$. Consider now the following configuration $\rho_0$: $prnt_r=prnt_b=\bot$, $level_r=level_b=mr$, $dist_r=dist_b=0$ and $prnt$, $level$, and $dist$  variables of other processes are arbitrary (see Figure \ref{fig:impTAstrong}, other variables may have arbitrary values but other variables of $b$ are identical to those of $r$).

Assume now that $b$ takes exactly the same actions as $r$ (if any) immediately after $r$ (note that $r\notin A_B^*$ and hence $prnt_r=\bot$, $level_r=mr$, and $dist_r=0$ still hold by closure and then $prnt_b=\bot$, $level_b=mr$, and $dist_r=0$ still hold too). Then, by symmetry of the execution and by convergence of $\mathcal{P}$ to $spec$, we can deduce that the system reaches in a finite time a configuration $\rho_1$ (see Figure \ref{fig:impTAstrong}) in which: $prnt_r=prnt_b=\bot$, $prnt_u=prnt_{u'}=r$, $prnt_v=prnt_{v'}=b$, $level_r=level_b=mr$, $level_u=level_{u'}=level_v=level_{v'}=m$, and $\forall v\in V,dist_v=legal\_dist_{prnt_v}$ (because this configuration is the only one in which all correct process $v$ satisfies $spec(v)$ when $prnt_r=prnt_b=\bot$ and $level_r=level_b=mr$ since $met(m,w')\prec m$). Note that $\rho_1$ is $A_B^*$-legitimate for $spec$ and $A_B^*$-stable (whatever $A_B^*$ is).

Assume now that $b$ behaves as a correct processor with respect to $\mathcal{P}$. Then, by convergence of $\mathcal{P}$ in a fault-free system starting from $\rho_1$ which is not legitimate (remember that a TA-strongly stabilizing algorithm is a special case of self-stabilizing algorithm), we can deduce that the system reach in a finite time a configuration $\rho_2$ (see Figure \ref{fig:impTAstrong}) in which: $prnt_r=\bot$, $prnt_u=prnt_{u'}=r$, $prnt_v=u$, $prnt_{v'}=u'$, $prnt_b=v$ (or $prnt_b=v'$), $level_r=mr$, $level_u=level_{u'}=m$ $level_v=level_{v'}=met(m,w')=m'$, $level_b=met(m',w)=m''$, and $\forall v\in V,dist_v=legal\_dist_{prnt_v}$. Note that processes $v$ and $v'$ modify their O-variables in the portion of execution between $\rho_1$ and $\rho_2$ and that $\rho_2$ is $A_B^*$-legitimate for $spec$ and $A_B^*$-stable (whatever $A_B^*$ is). Consequently, this portion of execution contains at least one $A_B^*$-TA-disruption (whatever $A_B^*$ is).

Assume now that the Byzantine process $b$ takes the following state: $prnt_b=\bot$ and $level_b=mr$. This step brings the system into configuration $\rho_3$ (see Figure \ref{fig:impTAstrong}). From this configuration, we can repeat the execution we constructed from $\rho_0$. By the same token, we obtain an execution of $\mathcal{P}$ which contains $c$-legitimate and $c$-stable configurations (see $\rho_1$) and an infinite number of $A_B^*$-TA-disruption (whatever $A_B^*$ is) which contradicts the $(t,A_B^*,1)$-TA-strong stabilization of $\mathcal{P}$.
\end{description}
\end{description}
\end{proof}

\begin{figure}[t]
\noindent \begin{centering} \include{impTAstrong}
  \par\end{centering}
 \caption{Configurations used in proof of Theorem \ref{th:impTAstrong}.}
\label{fig:impTAstrong}
\end{figure}

\section{Topology-Aware Strongly Stabilizing Protocol}\label{sec:protocol}

The goal of this section is to provide a $(t,S_B^*,n-1)$-TA strongly stabilizing protocol in order to match the lower bound on containment area provided by the Theorem \ref{th:impTAstrong}. If we focus on the protocol provided by \cite{DMT10ca} (which is $(S_B,n-1)$-TA strictly stabilizing), we can prove that this protocol does not satisfy our constraints since we have the following result.

\begin{theorem}\label{th:not2TAStrong}
Given a maximizable metric $\mathcal{M}=(M,W,mr,met,\prec)$, the protocol of \cite{DMT10ca} is not a $(t,S_B^*,2)$-TA strongly stabilizing protocol for maximum metric spanning tree construction with respect to $\mathcal{M}$ where $t$ is a given finite integer.
\end{theorem}

\begin{proof}
To prove this result, it is sufficient to construct an execution of the protocol of \cite{DMT10ca} for a given metric $\mathcal{M}$ which contains an infinite number of $S_B^*$-TA disruptions with two Byzantine processes.

Consider the shortest path metric $\mathcal{SP}$ defined above and the weighted system defined by Figure \ref{fig:CounterExample01} ($r$ denotes the root and $b_1$ and $b_2$ are two Byzantine processes). We recall that the protocol of \cite{DMT10ca} uses an upper bound $D$ on the length of any path of the tree and that the protocol is built in such a way that a process cannot choose as parent a neighbor with a $dist$ variable greater or equals to $D-1$. Here, we assume that $D=10$.

If we consider the initial configuration $\rho_1$ defined by Figure \ref{fig:CounterExample02}, we can state that processes $p_2$ and $p_3$ cannot modify their state as long as $b_1$ remains in its state. Moreover, $r$ and $p_1$ are never enabled by the protocol. In this way, it is possible to construct the following portion of execution $e_1$: $b_2$ modifies its level variable to 1. Then, $p_5$ and $p_4$ update their level variable to obtain configuration $\rho_2$ of Figure \ref{fig:CounterExample02}. Note that $e_1$ contains a $S_B^*$-TA disruption since $p_4$ modified one of its O-variables (namely, level) and $p_4\notin S_B^*$. From $\rho_2$, it is possible to construct the following portion of execution $e_2$: $b_2$ modifies its level variable to 0. Then, $p_5$ and $p_4$ update their level variable to obtain configuration $\rho_1$.

\begin{figure}
\noindent \begin{centering} \ifx\JPicScale\undefined\def\JPicScale{0.9}\fi
\unitlength \JPicScale mm
\begin{picture}(120,47.5)(0,0)
\linethickness{0.3mm}
\put(5,24.6){\circle{10}}

\linethickness{0.3mm}
\put(25,24.6){\circle{10}}

\linethickness{0.3mm}
\put(45,24.6){\circle{10}}

\linethickness{0.3mm}
\put(75,34.6){\circle{10}}

\linethickness{0.3mm}
\put(75,14.6){\circle{10}}

\linethickness{0.3mm}
\put(95,34.6){\circle{10}}

\linethickness{0.3mm}
\put(95,14.6){\circle{10}}

\linethickness{0.3mm}
\put(115,14.6){\circle{10}}

\put(95,34.6){\makebox(0,0)[cc]{$b_1$}}

\put(115,14.6){\makebox(0,0)[cc]{$b_2$}}

\put(5,24.6){\makebox(0,0)[cc]{$r$}}

\linethickness{0.3mm}
\put(10,24.6){\line(1,0){10}}
\put(10,24.6){\vector(-1,0){0.12}}
\linethickness{0.3mm}
\put(30,24.6){\line(1,0){10}}
\linethickness{0.3mm}
\multiput(50,24.6)(0.24,0.12){83}{\line(1,0){0.24}}
\put(70,34.6){\vector(2,1){0.12}}
\linethickness{0.3mm}
\multiput(50,24.6)(0.24,-0.12){83}{\line(1,0){0.24}}
\linethickness{0.3mm}
\put(80,34.6){\line(1,0){10}}
\put(90,34.6){\vector(1,0){0.12}}
\linethickness{0.3mm}
\put(80,14.6){\line(1,0){10}}
\put(90,14.6){\vector(1,0){0.12}}
\linethickness{0.3mm}
\put(100,14.6){\line(1,0){10}}
\put(110,14.6){\vector(1,0){0.12}}
\put(25,24.6){\makebox(0,0)[cc]{$p_1$}}

\put(45,24.6){\makebox(0,0)[cc]{$p_2$}}

\put(75,34.6){\makebox(0,0)[cc]{$p_3$}}

\put(75,14.6){\makebox(0,0)[cc]{$p_4$}}

\put(95,14.6){\makebox(0,0)[cc]{$p_5$}}

\textcolor{blue}{
\linethickness{0.3mm}
\qbezier(77.5,42.1)(79.23,41.62)(80.4,40.17)
\qbezier(80.4,40.17)(81.57,38.72)(82.5,37.1)
\qbezier(82.5,37.1)(83.46,35.33)(83.75,33.35)
\qbezier(83.75,33.35)(84.04,31.37)(85,29.6)
\qbezier(85,29.6)(85.95,28.04)(87.19,26.77)
\qbezier(87.19,26.77)(88.42,25.5)(90,24.6)
\qbezier(90,24.6)(92.42,23.53)(95.23,23.62)
\qbezier(95.23,23.62)(98.04,23.72)(100,22.1)
\qbezier(100,22.1)(101.92,20.14)(102.72,17.35)
\qbezier(102.72,17.35)(103.53,14.56)(102.5,12.1)
\qbezier(102.5,12.1)(101.19,9.46)(98.28,7.99)
\qbezier(98.28,7.99)(95.38,6.52)(92.5,7.1)
\qbezier(92.5,7.1)(89.63,8.07)(88.36,11.32)
\qbezier(88.36,11.32)(87.09,14.57)(85,17.1)
\qbezier(85,17.1)(82.66,19.28)(80,20.85)
\qbezier(80,20.85)(77.34,22.42)(75,24.6)
\qbezier(75,24.6)(72.57,26.83)(70.25,29.16)
\qbezier(70.25,29.16)(67.93,31.49)(67.5,34.6)
\qbezier(67.5,34.6)(67.41,36.04)(68.14,37.34)
\qbezier(68.14,37.34)(68.88,38.64)(70,39.6)
\qbezier(70,39.6)(71.57,40.98)(73.54,41.75)
\qbezier(73.54,41.75)(75.51,42.51)(77.5,42.1)
}

\textcolor{red}{
\linethickness{0.3mm}
\qbezier(72.5,47.1)(69.24,47.35)(66.15,46.59)
\qbezier(66.15,46.59)(63.06,45.83)(60,44.6)
\qbezier(60,44.6)(56.5,43.25)(53.41,41.41)
\qbezier(53.41,41.41)(50.32,39.57)(47.5,37.1)
\qbezier(47.5,37.1)(44.16,34.44)(40.97,31.51)
\qbezier(40.97,31.51)(37.79,28.58)(37.5,24.6)
\qbezier(37.5,24.6)(37.62,21.47)(39.99,19.05)
\qbezier(39.99,19.05)(42.36,16.64)(45,14.6)
\qbezier(45,14.6)(47.22,12.77)(49.73,11.59)
\qbezier(49.73,11.59)(52.23,10.41)(55,9.6)
\qbezier(55,9.6)(59.98,8.27)(64.95,8.17)
\qbezier(64.95,8.17)(69.92,8.08)(75,7.1)
\qbezier(75,7.1)(78.92,6.04)(82.46,4.31)
\qbezier(82.46,4.31)(86,2.58)(90,2.1)
\qbezier(90,2.1)(94.01,1.65)(98.07,1.85)
\qbezier(98.07,1.85)(102.13,2.06)(105,4.6)
\qbezier(105,4.6)(107.27,6.95)(107.73,10.37)
\qbezier(107.73,10.37)(108.19,13.79)(107.5,17.1)
\qbezier(107.5,17.1)(107.03,19.39)(105.72,21.32)
\qbezier(105.72,21.32)(104.41,23.25)(102.5,24.6)
\qbezier(102.5,24.6)(100.24,25.88)(97.45,25.79)
\qbezier(97.45,25.79)(94.66,25.69)(92.5,27.1)
\qbezier(92.5,27.1)(89.49,30.03)(88.66,34.48)
\qbezier(88.66,34.48)(87.83,38.94)(85,42.1)
\qbezier(85,42.1)(82.42,44.42)(79.2,45.61)
\qbezier(79.2,45.61)(75.99,46.8)(72.5,47.1)
}

\put(15,27.1){\makebox(0,0)[cc]{1}}

\put(35,27.1){\makebox(0,0)[cc]{1}}

\put(60,32.1){\makebox(0,0)[cc]{1}}

\put(60,17.1){\makebox(0,0)[cc]{0}}

\put(85,37.1){\makebox(0,0)[cc]{1}}

\put(85,12.1){\makebox(0,0)[cc]{1}}

\put(105,12.1){\makebox(0,0)[cc]{1}}

\textcolor{red}{
\put(65,24.6){\makebox(0,0)[cc]{$S_B$}}
}

\textcolor{blue}{
\put(82.5,24.6){\makebox(0,0)[cc]{$S_B^*$}}
}
\end{picture}
  \par\end{centering}
 \caption{System used in proof of Theorem \ref{th:not2TAStrong}.}
\label{fig:CounterExample01}
\end{figure}
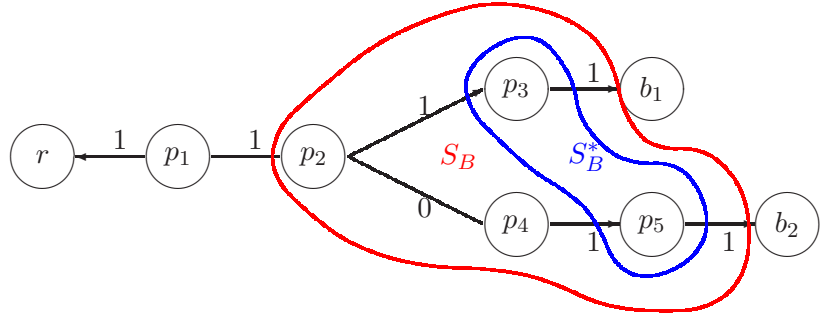

\begin{figure}
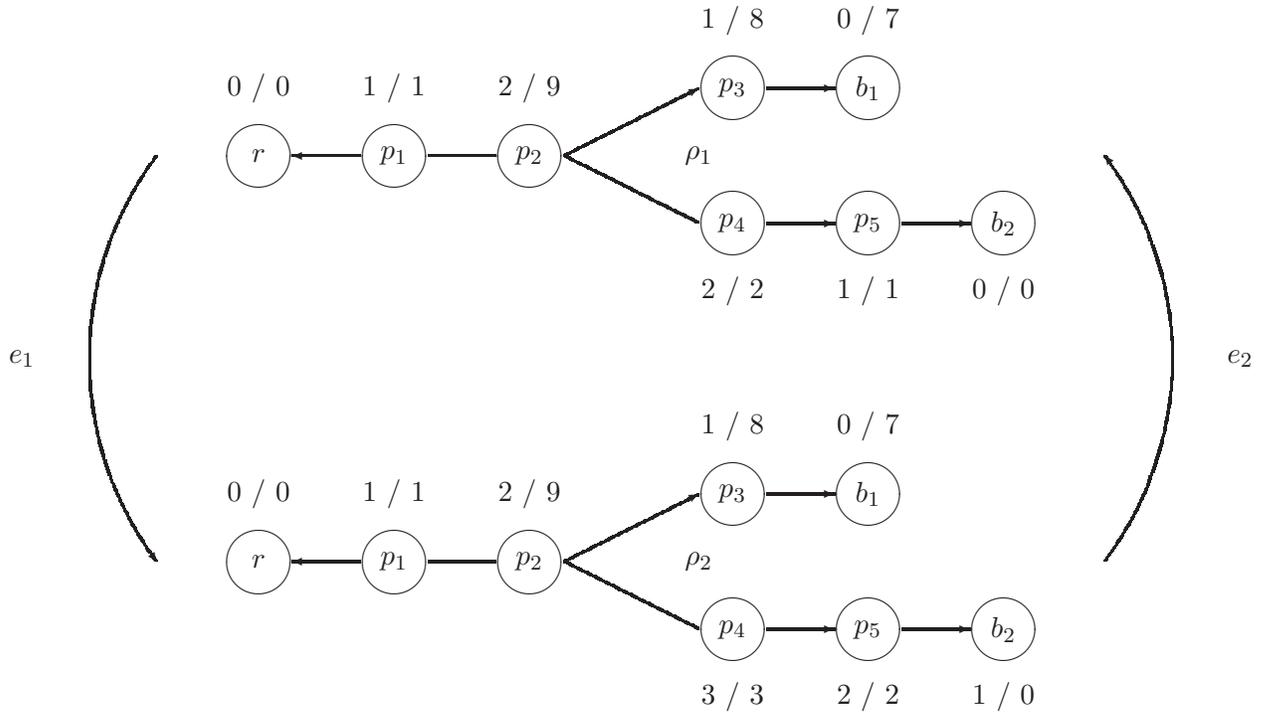

\noindent \begin{centering} \include{CounterExample02}
  \par\end{centering}
 \caption{Configurations used in proof of Theorem \ref{th:not2TAStrong} (for each process $v$, we use the notation $level_v$ / $dist_v$).}
\label{fig:CounterExample02}
\end{figure}

Consequently, it is possible to construct an infinite execution $e_1e_2e_1e_2\ldots$ starting from $\rho_1$ that contains an infinite number of $S_B^*$-TA disruptions with two Byzantine processes. This finishes the proof.

\end{proof}

\subsection{Presentation of the Protocol}

In contrast of Theorem \ref{th:not2TAStrong}, we provide in this paper a new protocol which is $(t,S_B^*,n-1)$-TA strongly stabilizing for maximum metric spanning tree construction. Our protocol needs a supplementary assumption on the system. We introduce the following definition.

\begin{definition}[Set of used metric values]
Given an assigned metric $\mathcal{AM}=(M,W,met,mr,\prec,wf)$ over a system $S$, the \emph{set of used metric values} of $\mathcal{AM}$ is defined as $M(S)=\{m\in M|\exists v\in V,(\mu(v,r)=m)\vee(\exists b\in B, \mu(v,b)=m)\}$.
\end{definition}

We assume that we always have $|M(S)|\geq 2$ (the necessity of this assumption is explained below). Nevertheless, note that the contrary case ($|M(S)|=1$) is possible if and only if the assigned metric is equivalent to $\mathcal{NC}$. As the protocol of \cite{DMT11j} performs $(t,0,n-1)$- strong stabilization with a finite $t$ for this metric, we can achieves the $(t,S_B^*,n-1)$-TA strong stabilization when $|M(S)|=1$ (since this implies that $S_B^*=\emptyset$). In this way, this assumption does not weaken the possibility result.

Although the protocol of \cite{DMT10ca} is not TA strongly stabilizing (see Theorem \ref{th:not2TAStrong}), our protocol borrows fundamental strategy from it. In this protocol, any process try to maximize its $level$ in the tree by choosing as its parent the neighbor that provide the best metric value. The key idea of this protocol is to use the distance variable (upper bounded by a given constant $D$) to detect and break cycles of process which has the same maximum metric. To achieve the TA strict stabilization, the protocol ensures a fair selection along the set of its neighbor with a round-robin order.

The possibility of infinite number of disruptions of the protocol of \cite{DMT10ca} mainly comes from the following fact: a Byzantine process can independently lie about its $level$ and its $dist$ variable. For example, a Byzantine process can provide a $level$ equals to $mr$ and a $dist$ arbitrarily large. In this way, it may lead a correct process of $S_B\setminus S_B^*$ to have a $dist$ variable equals to $D-1$ such that no other correct process can choose it as its parent (this rule is necessary to break cycle) but it cannot modify its state (this rule is only enabled when $dist$ is equals to $D$). Then, this process may always prevent some of its neighbors to join a $\mathcal{M}$-path connected to the root and hence allow another Byzantine process to perform an infinite number of disruptions.

It is why we modified the management of the $dist$ variable (note that others variables are managed exactly in the same way as in the protocol of \cite{DMT10ca}). In order to contain the effect of Byzantine process on $dist$ variables, each process that has a $level$ different from the one of its parent in the tree sets its $dist$ variable to $0$. In this way, a Byzantine process modifying its $dist$ variable can only affect correct process that have the same $level$. Consequently, in the case where $|M(S)|\geq 2$, we are ensured that correct processes of $S_B\setminus S_B^*$ cannot keep a $dist$ variable equals or greater than $D-1$ infinitely. Hence, a correct process of $S_B\setminus S_B^*$ cannot be disturbed infinitely often without joining a $\mathcal{M}$-path connected to the root.

We can see that the assumption $|M(S)|\geq 2$ is essential to perform the topology-aware strong stabilization. Indeed, in the case where $|M(S)|=1$, Byzantine processes can play exactly the scenario described above (in this case, our protocol is equivalent to the one of \cite{DMT10ca}).

The second modification we bring to the protocol of \cite{DMT10ca} follows. When a process has an inconsistent $dist$ variable with its parent, we allow it only to increase its $dist$ variable. If the process needs to decrease its $dist$ variable (when it has a strictly greater distance than its parent), then the process must change its parent. This rule allows us to bound the maximal number of steps of any process between two modifications of its parent (a Byzantine process cannot lead a correct one to infinitely often increase and decrease its distance without modifying its pointer).

Our protocol is formally described in Algorithm \ref{algo:max}.

\begin{algorithm}
\caption{$\mathcal{SSMAX}$, TA strongly stabilizing protocol for maximum metric tree construction.}\label{algo:max}
\scriptsize
\begin{description}
\item{Data:}~\\
$N_v$: totally ordered set of neighbors of $v$.\\
$D$: upper bound of the number of processes in a simple path.
\item{Variables:}~\\
$prnt_v\in\begin{cases}\{\bot\} \text{ if } v=r\\ N_v \text{ if } v\neq r\end{cases}$: pointer on the parent of $v$ in the tree.\\
$level_v\in\{m\in M|m\preceq mr\}$: metric of the node.\\
$dist_v\in\{0,\ldots,D\}$: hop counter.
\item{Functions:}~\\
For any subset $A\subseteq N_v$, $choose_v(A)$ returns the first element of $A$ which is bigger than $prnt_v$ (in a round-robin fashion).\\
$current\_dist_v()=\left\{\begin{array}{l}
0 \mbox{ if } level_{prnt_v}\neq level_v\\
min(dist_{prnt_v}+1,D) \mbox{ if } level_{prnt_v}=level_v
\end{array}\right.$
\item{Rules:}~\\
$\boldsymbol{(R_r)}::(v=r)\wedge((level_v\neq mr)\vee(dist_v\neq 0))\longrightarrow level_v:=mr;~dist_v:=0$

$\boldsymbol{(R_1)}::(v\neq r)\wedge(prnt_v\in N_v)\wedge((dist_v<current\_dist_v())\vee(level_v\neq met(level_{prnt_v},w_{v,prnt_v})))$\\ $~~~~~~~~~~~~~\longrightarrow level_v:=met(level_{prnt_v},w_{v,prnt_v});~dist_v:=current\_dist_v()$

$\boldsymbol{(R_2)}::(v\neq r)\wedge((dist_v=D)\vee(dist_v>current\_dist_v()))\wedge(\exists u \in N_v, dist_u<D-1)$\\$~~~~~~~~~~~~~\longrightarrow prnt_v:=choose_v(\{u\in N_v|dist_v<D-1\});~level_v:=met(level_{prnt_v},w_{v,prnt_v});~dist_v:=current\_dist_v()$

$\boldsymbol{(R_3)}::(v\neq r)\wedge(\exists u\in N_v,(dist_u<D-1)\wedge(level_v\prec met(level_u,w_{u,v})))$\\
$~~~~~~~~~~~\longrightarrow prnt_v:=choose_v\Bigg(\Bigg\{u\in N_v\Big|(level_u<D-1)\wedge(met(level_u,w_{u,v})=\underset{q\in N_v/level_q<D-1}{max_\prec}\{met(level_q,w_{q,v})\})\Bigg\}\Bigg);$\\
$~~~~~~~~~~~~~~~level_v:=met(level_{prnt_v},w_{prnt_v,v});~dist_v:=current\_dist_v()$
\end{description}
\end{algorithm}
\normalsize

\subsection{Proof of the $(S_B,n-1)$-TA Strict Stabilization for $spec$}

This proof is similar to the one of \cite{DMT10ca} but we must modify it to take in account modifications of the protocol. In \cite{DMT10ca}, we proved the following useful property about maximizable metrics.

\begin{lemma}\label{lem:goodProperty}
For any process $v\in V$, we have:
\[\forall u\in N_v,met\left(\underset{p\in B\cup\{r\}}{max_\prec}\{\mu(u,p)\},w_{u,v}\right)\preceq\underset{p\in B\cup\{r\}}{max_\prec}\{\mu(v,p)\}\]
\end{lemma}

Given a configuration $\rho\in C$ and a metric value $m\in M$, let us define the following predicate: 
\[IM_m(\rho)\equiv \forall v\in V,level_v\preceq max_\prec\left\{m,\underset{u\in B\cup\{r\}}{max_\prec}\{\mu(v,u)\}\right\}\]

\begin{lemma}\label{lem:Imclosed}
For any metric value $m\in M$, the predicate $IM_m$ is closed by actions of $\mathcal{SSMAX}$.
\end{lemma}

\begin{proof}
Let $m$ be a metric value ($m\in M$). Let $\rho\in C$ be a configuration such that $IM_m(\rho)=true$ and $\rho'\in C$ be a configuration such that $\rho \stackrel{R}{\mapsto} \rho'$ is a step of $\mathcal{SSMAX}$.

If the root process $r\in R$ (respectively a Byzantine process $b\in R$), then we have $level_r=mr$ (respectively $level_b\preceq mr$) in $\rho'$ by construction of $\boldsymbol{(R_r)}$ (respectively by definition of $level_b$). Hence, $level_r\preceq max_\prec\left\{m, \underset{u\in B\cup\{r\}}{max_\prec}\{\mu(r,u)\}\right\}=mr$ (respectively $level_b\preceq max_\prec\left\{m,\underset{u\in B\cup\{r\}}{max_\prec}\{\mu(b,u)\}\right\}=mr$).

If a correct process $v\in R$ with $v\neq r$, then there exists a neighbor $p$ of $v$ such that $level_p\preceq max_\prec\left\{m,\underset{u\in B\cup\{r\}}{max_\prec}\{\mu(p,u)\}\right\}$ in $\rho$ (since $IM_m(\rho)=true$) and $prnt_v=p$ and $level_v=met(level_p,$ $w_{v,p})$ in $\rho'$ (since $v$ is activated during this step).

If we apply the Lemma \ref{lem:goodProperty} to $met$ and to neighbor $p$, we obtain the following property:
\[met\left(\underset{u\in B\cup\{r\}}{max_\prec}\{\mu(p,u)\},w_{v,p}\right)\preceq\underset{u\in B\cup\{r\}}{max_\prec}\{\mu(v,u)\}\]  

Consequently, we obtain that, in $\rho'$:
\[\begin{array}{rcll}
level_v & = & met(level_p,w_{v,p})&\\
& \preceq & met\left(max_\prec\left\{m,\underset{u\in B\cup\{r\}}{max_\prec}\{\mu(p,u)\}\right\},w_{v,p}\right)&\text{ by monotonicity of }\mathcal{M}\\
& \preceq & max_\prec\left\{met(m,w_{v,p}),met\left(\underset{u\in B\cup\{r\}}{max_\prec}\{\mu(p,u)\},w_{v,p}\right) \right\}&\\
& \preceq & max_\prec\left\{m,\underset{u\in B\cup\{r\}}{max_\prec}\{\mu(v,u)\}\right\}&\text{ since } met(m,w_{v,p})\preceq m
\end{array}\]

We can deduce that $IM_m(\rho')=true$, that concludes the proof.
\end{proof}

Given an assigned metric to a system $G$, we can observe that the set of metrics value $M$ is finite and that we can label elements of $M$ by $m_0=mr,m_1,\ldots,m_k$ in a way such that $\forall i\in\{0,\ldots,k-1\},m_{i+1}\prec m_i$.

We introduce the following notations:
\[\begin{array}{rrcl}
\forall m_i\in M, & P_{m_i} & = & \big\{v\in V\setminus S_B\big| \mu(v,r)=m_i\big\}\\
\forall m_i\in M, & V_{m_i} & = & \underset{j=0}{\overset{i}{\bigcup}}P_{m_j}\\
\forall m_i\in M, & I_{m_i} & = & \big\{v\in V\big|\underset{u\in B\cup\{r\}}{max_\prec}\{\mu(v,u)\}\prec m_i \big\}\\
\forall m_i\in M, & \mathcal{LC}_{m_i} & = & \big\{\rho\in\mathcal{C}\big|(\forall v\in V_{m_i}, spec(v))\wedge(IM_{m_i}(\rho))\big\}\\
 & \mathcal{LC} & = & \mathcal{LC}_{m_k}
\end{array}\]

\begin{lemma}\label{lem:LCmiclosed}
For any $m_i\in M$, the set $\mathcal{LC}_{m_i}$ is closed by actions of $\mathcal{SSMAX}$.
\end{lemma}

\begin{proof}
Let $m_i$ be a metric value from $M$ and $\rho$ be a configuration of $\mathcal{LC}_{m_i}$. By construction, any process $v\in V_{m_i}$ satisfies $spec(v)$ in $\rho$. 

In particular, the root process satisfies: $prnt_r=\bot$, $level_r=mr$, and $dist_r=0$. By construction of $\mathcal{SSMAX}$, $r$ is not enabled and then never modifies its O-variables (since the guard of the rule of $r$ does not involve the state of its neighbors). 

In the same way, any process $v\in V_{m_i}$ satisfies: $prnt_v\in N_v$, $level_v=met(level_{prnt_v},$ $w_{prnt_v,v})$, $dist_v=legal\_dist_{prnt_v}$, and $level_v=\underset{u\in N_v}{max_\prec}\{met(level_u,w_{u,v})\}$. Note that, as $v\in V_{m_i}$ and $spec(v)$ holds in $\rho$, we have: $level_v=\mu(v,r)=\underset{p\in B\cup\{r\}}{max_\prec}\{\mu(v,p)\}$ and $dist_v\leq D-1$ by construction of $D$. Hence, process $v$ is not enabled in $\rho$.

Assume that there exists a process $v\in V_{m_i}$ that takes a step $\rho' \stackrel{R}{\mapsto} \rho''$ in an execution starting from $\rho$ (without loss of generality, assume that $v$ is the first process of $v\in V_{m_i}$ that takes a step in this execution). Then, we know that $v\neq r$. This activation implies that a neighbor $u\notin V_{m_i}$ (since $v$ is the first process of $V_{m_i}$ to take a step) of $v$ modified its $level_u$ variable to a metric value $m\in M$ such that $level_v\prec met(m,w_{u,v})$ in $\rho'$ (note that O-variables of $v$ and $prnt_v$ remain consistent since $v$ is the first process to take a step in this execution).

Hence, we have $level_v=\underset{p\in B\cup\{r\}}{max_\prec}\{\mu(v,p)\}=\mu(v,r)$ (since $spec(v)$ holds), $level_v\prec met(m,w_{u,v})$ (since $u$ causes an action of $v$), and $m_i\preceq level_v$ (since $v\in V_{m_i}$ and $level_v=\mu(v,r)$). Moreover, the closure of $IM_{m_i}$ (established in Lemma \ref{lem:Imclosed}) ensures us that $m=level_u\preceq max_\prec\left\{m_i,\underset{p\in B\cup\{r\}}{max_\prec}\{\mu(u,p)\}\right\}$. Let us study the two following cases:
\begin{description}
\item[Case 1:] $max_\prec\left\{m_i,\underset{p\in B\cup\{r\}}{max_\prec}\{\mu(u,p)\}\right\}=m_i$.\\
We have then $m\preceq m_i$. As the boundedness of $\mathcal{M}$ ensures that $met(m,w_{u,v})\preceq m$, we can conclude that $level_v\prec met(m,w_{u,v})\preceq m\preceq m_i\preceq level_v$, that is absurd.
\item[Case 2:] $max_\prec\left\{m_i,\underset{p\in B\cup\{r\}}{max_\prec}\{\mu(u,p)\}\right\}=\underset{p\in B\cup\{r\}}{max_\prec}\{\mu(u,p)\}$.\\
We have then $m\preceq \underset{p\in B\cup\{r\}}{max_\prec}\{\mu(u,p)\}$. By monotonicity of $\mathcal{M}$, we can deduce that $met(m,w_{u,v})\preceq met(\underset{p\in B\cup\{r\}}{max_\prec}\{\mu(u,p)\},w_{u,v})$. Consequently, we obtain that $\underset{p\in B\cup\{r\}}{max_\prec}\{\mu(v,p)\}\prec met(\underset{p\in B\cup\{r\}}{max_\prec}\{\mu(u,p)\},w_{u,v})$. This is contradictory with the result of Lemma \ref{lem:goodProperty}.
\end{description}

In conclusion, any process $v\in V_{m_i}$ takes no step in any execution starting from $\rho$ and then always satisfies $spec(v)$. Then, the closure of $IM_B$ (established in Lemma \ref{lem:Imclosed}) concludes the proof.
\end{proof}

\begin{lemma}\label{lem:SBTAcontainedMax}
Any configuration of $\mathcal{LC}$ is $(S_B,n-1)$-TA contained for $spec$.
\end{lemma}

\begin{proof}
This is a direct application of the Lemma \ref{lem:LCmiclosed} to $\mathcal{LC}=\mathcal{LC}_{m_k}$.
\end{proof}

\begin{lemma}\label{lem:CtoLCmr}
Starting from any configuration of $\mathcal{C}$, any execution of $\mathcal{SSMAX}$ reaches in a finite time a configuration of $\mathcal{LC}_{mr}$.
\end{lemma}

\begin{proof}
Let $\rho$ be an arbitrary configuration. Then, it is obvious that $IM_{mr}(\rho)$ is satisfied. By closure of $IM_{mr}$ (proved in Lemma \ref{lem:Imclosed}), we know that $IM_{mr}$ remains satisfied in any execution starting from $\rho$.

If $r$ does not satisfy $spec(r)$ in $\rho$, then $r$ is continuously enabled. Since the scheduling is strongly fair, $r$ is activated in a finite time and then $r$ satisfies $spec(r)$ in a finite time. Denote by $\rho'$ the first configuration in which $spec(r)$ holds. Note that $r$ takes no step in any execution starting from $\rho'$.

The boundedness of $\mathcal{M}$ implies that $P_{mr}$ induces a connected subsystem. If $P_{mr}=\{r\}$, then we proved that $\rho'\in\mathcal{LC}_{mr}$ and we have the result.

Otherwise, observe that, for any configuration of an execution starting from $\rho'$, if all processes of $P_{mr}$ are not enabled, then all processes $v$ of $P_{mr}$ satisfy $spec(v)$. Assume now that there exists an execution $e$ starting from $\rho'$ in which some processes of $P_{mr}$ take infinitely many steps. By construction, at least one of these processes (note it $v$) has a neighbor $u$ which takes only a finite number of steps in $e$ (recall that $P_{mr}$ induces a connected subsystem and that $r$ takes no step in $e$). After $u$ takes its last step of $e$, we can observe that $level_u=mr$ and $dist_u<D-1$ (otherwise, $u$ is activated in a finite time that contradicts its construction). 

As $v$ can execute consequently $\boldsymbol{(R_1)}$ only a finite number of times (since the incrementation of $dist_v$ is bounded by $D$), we can deduce that $v$ executes $\boldsymbol{(R_2)}$ or $\boldsymbol{(R_3)}$ infinitely often. In both cases, $u$ belongs to the set which is the parameter of function $choose$. By the fairness of this function, we can deduce that $prnt_v=u$ in a finite time in $e$. Then, the construction of $u$ implies that $v$ is never enabled in the sequel of $e$. This is contradictory with the construction of $e$.

Consequently, any execution starting from $\rho'$ reaches in a finite time a configuration such that all processes of $P_{mr}$ are not enabled. We can deduce that this configuration belongs to $\mathcal{LC}_{mr}$, that ends the proof.
\end{proof}

\begin{lemma}\label{lem:LCmitodistD}
For any $m_i\in M$ and for any configuration $\rho\in \mathcal{LC}_{m_i}$, any execution of $\mathcal{SSMAX}$ starting from $\rho$ reaches in a finite time a configuration such that:
\[\forall v\in I_{m_i},level_v=m_i\Rightarrow dist_v=D\]
\end{lemma}

\begin{proof}
Let $m_i$ be an arbitrary metric value of $M$ and $\rho_0$ be an arbitrary configuration of $\mathcal{LC}_{m_i}$. Let $e=\rho_0,\rho_1,\ldots$ be an execution starting from $\rho_0$.

Note that $\rho_0$ satisfies $IM_{m_i}$ by construction. Hence, we have $\forall v\in I_{m_i},level_v\preceq m_i$. The closure  of $IM_{m_i}$ (proved in Lemma \ref{lem:Imclosed}) ensures us that this property is satisfied in any configuration of $e$. 

If any process $v\in I_{m_i}$ satisfies $level_v\prec m_i$ in $\rho_0$, then the result is obvious. Otherwise, we define the following variant function. For any configuration $\rho_j$ of $e$, we denote by $A_j$ the set of processes $v$ of $I_{m_i}$ such that $level_v=m_i$ in $\rho_j$. Then, we define $f(\rho_j)=\underset{v\in A_j}{min}\{dist_v\}$. We will prove the result by showing that there exists an integer $k$ such that $f(\rho_k)=D$.

First, if a process $v$ joins $A_j$ (that is, $v\notin A_{j-1}$ but $v\in A_j$), then it takes a distance value greater or equals to $f(\rho_{j-1})+1$ by construction of the protocol. We can deduce that any process that joins $A_j$ does not decrease $f$. Moreover, the construction of the protocol implies that a process $v$ such that $v\in A_j$ and $v\in A_{j+1}$ can not decrease its distance value in the step $\rho_j\mapsto \rho_{j+1}$.

Then, consider for a given configuration $\rho_j$ a process $v\in A_j$ such that $dist_v=f(\rho_j)<D$. We claim that $v$ is enabled in $\rho_j$ and that the execution of the enabled rule either increases strictly $dist_v$ or removes $v$ from $A_{j+1}$. We distinguish the following cases:

\begin{description}
\item[Case 1:] $level_v=met(level_{prnt_v},w_{v,prnt_v})$\\
The fact that $v\in I_{m_i}$, the boundedness of $\mathcal{M}$ and the closure of $IM_{m_i}$ imply that $prnt_v\in A_j$ (and, hence that $level_{prnt_v}=m_i$). Then, by construction of $f(\rho_j)$, we know that $dist_{prnt_v}\geq f(\rho_j)=dist_v$. Hence, we have $dist_v< dist_{prnt_v}+1$ in $\rho_j$. Then, $v$ is enabled by $\boldsymbol{(R_1)}$ in $\rho_j$ and $dist_v$ increases of at least 1 during the step $\rho_j\mapsto \rho_{j+1}$ if this rule is executed.

\item[Case 2:] $level_v\neq met(level_{prnt_v},w_{v,prnt_v})$\\
Assume that $v$ is activated by $\boldsymbol{(R_2)}$ or $\boldsymbol{(R_3)}$ during the step $\rho_j\mapsto \rho_{j+1}$. If $v$ does not belong to $A_{j+1}$ (if $level_v\neq m_i$ in $\rho_{j+1}$), the claim is satisfied. In the contrary case ($v$ belongs to $A_{j+1}$), we know that $level_v=m_i$ in $\rho_{j+1}$. The boundedness of $\mathcal{M}$ and the closure of $IM_{m_i}$ imply that $level_{prnt_v}=m_i$ in $\rho_{j+1}$. We can conclude that $dist_v$ increases of at least 1 during the step $\rho_j\mapsto \rho_{j+1}$ since the new parent of $v$ has a distance greater than $f(\rho_j)$ by construction of $A_{j+1}$.

Otherwise, we know that the rule $\boldsymbol{(R_1)}$ is enabled for $v$ in $\rho_j$. If this rule is executed during the step $\rho_j\mapsto \rho_{j+1}$, one of the two following sub cases appears.
\begin{description}
\item[Case 2.1:] $met(level_{prnt_v},w_{v,prnt_v})\prec m_i$ in $\rho_j$.\\
Then, $v$ does not belong to $A_{j+1}$ by definition. 
\item[Case 2.2:] $met(level_{prnt_v},w_{v,prnt_v})=m_i$ in $\rho_j$.\\
Remind that the closure of $IM_{m_i}$ implies then that $level_{prnt_v}=m_i$. By construction of $f(\rho_j)$, we have $dist_{prnt_v}\geq f(\rho_j)$ in $\rho_j$. Then, we can see that $dist_v$ increases of at least 1 during the step $\rho_j\mapsto \rho_{j+1}$.
\end{description}
\end{description}

In all cases, $v$ is enabled (at least by $\boldsymbol{(R_1)}$) in $\rho_j$ and the execution of the enabled rule either increases strictly $dist_v$ or removes $v$ from $A_{j+1}$.

As $I_{m_i}$ is finite and the scheduling is strongly fair, we can deduce that $f$ increases in a finite time in any execution starting from $\rho_j$. By repeating the argument at most $D$ times, we can deduce that $e$ contains a configuration $\rho_k$ such that $f(\rho_k)=D$, that shows the result.
\end{proof}

\begin{lemma}\label{lem:distDtolevelmi}
For any $m_i\in M$ and for any configuration $\rho\in \mathcal{LC}_{m_i}$ such that
$\forall v\in I_{m_i},level_v=m_i\Rightarrow dist_v=D$, any execution of $\mathcal{SSMAX}$ starting from $\rho$ reaches in a finite time a configuration such that:
\[\forall v\in I_{m_i},level_v\prec m_i\]
\end{lemma}

\begin{proof}
Let $m_i\in M$ be an arbitrary metric value and $\rho_0$ be a configuration of $\mathcal{LC}_{m_i}$ such that $\forall v\in I_{m_i},level_v=m_i\Rightarrow dist_v=D$. Let $e=\rho_0,\rho_1,\ldots$ be an arbitrary execution starting from $\rho_0$.

For any configuration $\rho_j$ of $e$, let us denote $E_{\rho_j}=\{v\in I_{m_i}|level_v=m_i\}$. By the closure of $IM_{m_i}$ (which holds by definition in $\rho_0$) established in Lemma \ref{lem:Imclosed}, we obtain the result if there exists a configuration $\rho_j$ of $e$ such that $E_{\rho_j}=\emptyset$.

If there exist some processes $v\in I_{m_i}\setminus E_{\rho_0}$ (and hence $level_v\prec m_i$) such that $prnt_v\in E_{\rho_0}$ and $met(level_{prnt_v},w_{v,prnt_v})=m_i$ in $\rho_0$, then we can observe that these processes are continuously enabled by $\boldsymbol{(R_1)}$. As the scheduling is strongly fair, $v$ activates this rule in a finite time and then, $level_v=m_i$ and $dist_v=D$. In other words, $v$ joins $E_{\rho_l}$ for a given integer $l$. We can conclude that there exists an integer $k$ such that the following property $\boldsymbol{(P)}$ holds: for any $v\in I_{m_i}\setminus E_{\rho_0}$, either $prnt_v\notin E_{\rho_k}$ or $met(level_{prnt_v},w_{v,prnt_v})\prec m_i$.

Then, we prove that, for any integer $j\geq k$, we have $E_{\rho_{j+1}}\subseteq E_{\rho_j}$. For the sake of contradiction, assume that there exists an integer $j\geq k$ and a process $v\in I_{m_i}$ such that $v\in E_{\rho_{j+1}}$ and $v\notin E_{\rho_j}$. Without loss of generality, assume that $j$ is the smallest integer which performs these properties. Let us study the following cases:
\begin{description}
\item[Case 1:] $v$ activates $\boldsymbol{(R_1)}$ during the step $\rho_j\mapsto \rho_{j+1}$.\\
Note that the property $\boldsymbol{(P)}$ still holds in $\rho_j$ by the construction of $j$. Hence, we know that $prnt_v\notin E_{\rho_j}$ in $\rho_j$. But in this case, we have: $level_{prnt_v}\prec m_i$. The boundedness of $\mathcal{M}$ implies that $level_v=met(level_{prnt_v},w_{v,prnt_v})\prec m_i$ in $\rho_{j+1}$ that contradicts the fact that $v\in E_{\rho_{j+1}}$.
\item[Case 2:] $v$ activates either $\boldsymbol{(R_2)}$ or $\boldsymbol{(R_3)}$ during the step $\rho_j\mapsto \rho_{j+1}$.\\
That implies $v$ chooses a new parent which has a distance smaller than $D-1$ in $\rho_j$. This implies that this new parent does not belongs to $E_{\rho_j}$. Then, we have $level_{prnt_v}\prec m_i$. The boundedness of $\mathcal{M}$ implies that $level_v=met(level_{prnt_v},w_{v,prnt_v})\prec m_i$ in $\rho_{j+1}$ that contradicts the fact that $v\in E_{\rho_{j+1}}$.
\end{description}
In the two cases, our claim is satisfied. In other words, there exists a point of the execution (namely $\rho_k$) afterwards the set $E$ cannot grow (this implies that, if a process leaves the set $E$, it is a definitive leaving).

Assume now that there exists a step $\rho_j\mapsto \rho_{j+1}$ (with $j\geq k$) such that a process $v\in E_{\rho_j}$ is activated. Observe that the closure of $IM_{m_i}$ implies that $v$ can not be activated by the rule $\boldsymbol{(R_3)}$. If $v$ activates $\boldsymbol{(R_1)}$ during this step, then $v$ modifies its level during this step (otherwise, we have a contradiction with the fact that $level_{prnt_v}=m_i\Rightarrow dist_v=D$). The closure of $IM_{m_i}$ implies that $v$ leaves the set $E$ during this step. If $v$ activates $\boldsymbol{(R_2)}$ during this step, then $v$ chooses a new parent which has a distance smaller than $D-1$ in $\rho_j$. This implies that this new parent does not belongs to $E_{\rho_j}$. Then, we have $level_{prnt_v}\prec m_i$. The boundedness of $\mathcal{M}$ implies that $level_v\prec m_i$ in $\rho_{j+1}$. In other words, if a process of $E_{\rho_j}$ is activated during the step $\rho_j\mapsto \rho_{j+1}$, then it satisfies $v\notin E_{\rho_{j+1}}$.

Finally, observe that the construction of the protocol and the construction of the bound $D$ ensures us that any process $v\in I_{m_i}$ such that $dist_v=D$ is activated in a finite time. In conclusion, we obtain that there exists an integer $j$ such that $E_{\rho_j}=\emptyset$, that implies the result.
\end{proof}

\begin{lemma}\label{lem:CmitoIMmi+1}
For any $m_i\in M$ and for any configuration $\rho\in \mathcal{LC}_{m_i}$, any execution of $\mathcal{SSMAX}$ starting from $\rho$ reaches in a finite time a configuration $\rho'$ such that $IM_{m_{i+1}}$ holds.
\end{lemma}

\begin{proof}
This result is a direct consequence of Lemmas \ref{lem:LCmitodistD} and \ref{lem:distDtolevelmi}.
\end{proof}

\begin{lemma}\label{lem:LCmitoLCmi+1}
For any $m_i\in M$ and for any configuration $\rho\in \mathcal{LC}_{m_i}$, any execution of $\mathcal{SSMAX}$ starting from $\rho$ reaches in a finite time a configuration of $\mathcal{LC}_{m_{i+1}}$.
\end{lemma}

\begin{proof}
Let $m_i$ be a metric value of $M$ and $\rho$ be an arbitrary configuration of $\mathcal{LC}_{m_i}$. We know by Lemma \ref{lem:CmitoIMmi+1} that any execution starting from $\rho$ reaches in a finite time a configuration $\rho'$ such that $IM_{m_{i+1}}$ holds. By closure of $IM_{m_{i+1}}$ and of $\mathcal{LC}_{m_i}$ (established respectively in Lemma \ref{lem:Imclosed} and \ref{lem:LCmiclosed}), we know that any configuration of any execution starting from $\rho'$ belongs to $\mathcal{LC}_{m_i}$ and satisfies $IM_{m_{i+1}}$.

We know that $V_{m_i}\neq\emptyset$ since $r\in V_{m_i}$ for any $i\geq 0$. Remind that $V_{m_{i+1}}$ is connected by the boundedness of $\mathcal{M}$. Then, we know that there exists at least one process $p$ of $P_{m_{i+1}}$ which has a neighbor $q$ in $V_{m_{i}}$ such that $\mu(p,r)=met(\mu(q,r),w_{p,q})$. Moreover, Lemma \ref{lem:LCmiclosed} ensures us that any process of $V_{m_i}$ takes no step in any executions starting from $\rho'$.

Observe that, for any configuration of an execution starting from $\rho'$, if any process of $P_{m_{i+1}}$ is not enabled, then all processes $v$ of $P_{m_{i+1}}$ satisfy $spec(v)$. Assume now that there exists an execution $e$ starting from $\rho'$ in which some processes of $P_{m_{i+1}}$ take infinitely many steps. By construction, at least one of these processes (note it $v$) has a neighbor $u$ such that $\mu(v,r)=met(\mu(u,r),w_{v,u})$ which takes only a finite number of steps in $e$ (recall the construction of $p$). After $u$ takes its last step of $e$, we can observe that $level_u=\mu(u,r)$ and $dist_u<D-1$ (otherwise, $u$ is activated in a finite time that contradicts its construction). 

As $v$ can execute consequently $\boldsymbol{(R_1)}$ only a finite number of times (since the incrementation of $dist_v$ is bounded by $D$), we can deduce that $v$ executes $\boldsymbol{(R_2)}$ or $\boldsymbol{(R_3)}$ infinitely often. In both cases, $u$ belongs to the set which is the parameter of function $choose$ (remind that $IM_{m_{i+1}}$ is satisfied and that $u$ has the better possible metric among $v$'s neighbors). By the construction of this function, we can deduce that $prnt_v=u$ in a finite time in $e$. Then, the construction of $u$ implies that $v$ is never enabled in the sequel of $e$. This is contradictory with the construction of $e$.

Consequently, any execution starting from $\rho'$ reaches in a finite time a configuration such that all processes of $P_{m_{i+1}}$ are not enabled. We can deduce that this configuration belongs to $\mathcal{LC}_{m_{i+1}}$, that ends the proof.
\end{proof}

\begin{lemma}\label{lem:convergenceLCMax}
Starting from any configuration, any execution of $\mathcal{SSMAX}$ reaches a configuration of $\mathcal{LC}$ in a finite time.
\end{lemma}

\begin{proof}
Let $\rho$ be an arbitrary configuration. We know by Lemma \ref{lem:CtoLCmr} that any execution starting from $\rho$ reaches in a finite time a configuration of $\mathcal{LC}_{mr}=\mathcal{LC}_{m_0}$. Then, we can apply at most $k$ times the result of Lemma \ref{lem:LCmitoLCmi+1} to obtain that any execution starting from $\rho$ reaches in a finite time a configuration of $\mathcal{LC}_{m_k}=\mathcal{LC}$, that proves the result.
\end{proof}

\begin{theorem}\label{th:SSMAXstrict}
$\mathcal{SSMAX}$ is a $(S_B,n-1)$-TA-strictly stabilizing protocol for $spec$.
\end{theorem}

\begin{proof}
This result is a direct consequence of Lemmas \ref{lem:SBTAcontainedMax} and \ref{lem:convergenceLCMax}.
\end{proof}

\subsection{Proof of the $(t,S_B^*,n-1)$-TA Strong Stabilization for $spec$}

Let be $E_B=S_B\setminus S_B^*$ (\emph{i.e.} $E_B$ is the set of process $v$ such that $\mu(v,r)=\underset{b\in B}{max}\{\mu(v,b)\}$). Note that the subsytem induced by $E_B$ may have several connected components. In the following, we use the following notations: $E_B=\{E_B^1,\ldots,E_B^\ell\}$ where each $E_B^i$ ($i\in\{0,\ldots,\ell\}$) is a subset of $E_B$ inducing a maximal connected component, $\delta(E_B^i)$ ($i\in\{0,\ldots,\ell\}$) is the diameter of the subsystem induced by $E_B^i$, and $\delta=\underset{i\in\{0,\ldots,\ell\}}{max}\{\delta(E_B^i)\}$. When $a$ and $b$ are two integers, we define the following function: $\Pi(a,b)=\frac{a^{b+1}-1}{a-1}$.

\begin{lemma}\label{lem:LCnbactions}
If $\rho$ is a configuration of $\mathcal{LC}$, then any process $v\in E_B$ is activated at most $\Pi(k,\delta)\Delta D$ times in any execution starting from $\rho$.
\end{lemma}

\begin{proof}
Let $\rho$ be a configuration of $\mathcal{LC}$ and $e$ be an execution starting from $\rho$. Let $p$ be a process of $E_B^i$ ($i\in\{0,\ldots,\ell\}$) such that there exists a neighbor $q$ which satisfies $q\in V\setminus S_B$ and $\mu(p,r)=met(\mu(q,r),w_{p,q})$ (such a process exists by construction of $E_B^i$). We are going to prove by induction on $d$ the following property:

$\boldsymbol{(P_d)}$: if $v$ is a process of $E_B^i$ such that $d_{E_B^i}(p,v)=d$ (where $d_{E_B^i}$ denotes the distance in the subsystem induced by $E_B^i$), then $v$ executes at most $\Pi(k,d)\Delta D$ actions in $e$.

\begin{description}
\item[Initialization:] $d=0$.\\
This implies that $v=p$. Then, by construction, there exists a neighbor $q$ which satisfies $q\in V\setminus S_B$ and $\mu(p,r)=met(\mu(q,r),w_{p,q})$. As $\rho\in\mathcal{LC}$, Lemma \ref{lem:SBTAcontainedMax} ensures us that $level_q=\mu(q,r)$ and $dist_q<D-1$ in any configuration of $e$. Then, the boundedness of $\mathcal{M}$ implies that $q$ belongs to the set which is parameter to the macro $choose$ at any execution of rules $\boldsymbol{(R_2)}$ or $\boldsymbol{(R_3)}$ by $p$. Consequently, $p$ executes at most $\Delta$ times rules $\boldsymbol{(R_2)}$ and $\boldsymbol{(R_3)}$ in $e$ before choosing $q$ as its parent. Moreover, note that $p$ can execute rule $\boldsymbol{(R_1)}$ at most $D$ times between two consecutive executions of rules $\boldsymbol{(R_2)}$ and $\boldsymbol{(R_3)}$ (because $\boldsymbol{(R_1)}$ only increases $dist_p$ which is bounded by $D$). Consequently, $p$ executes at most $\Delta D$ actions before choosing $q$ as its parent.

By Lemma \ref{lem:SBTAcontainedMax}, we know that $q$ takes no action in $e$. Once $p$ chooses $q$ as its parent, its state is consistent with the one of $q$ (by construction of rules $\boldsymbol{(R_2)}$ and $\boldsymbol{(R_3)}$). Hence, $p$ is never enabled after choosing $q$ as its parent. Consequently, we obtain that $p$ takes at most $\Delta D$ actions in $e$, that proves $\boldsymbol{(P_0)}$.

\item[Induction:] $d>0$ and $\boldsymbol{(P_{d-1})}$ is true.\\
Let $v$ be a process of $E_B^i$ such that $d_{E_B^i}(p,v)=d$. By construction, there exists a neighbor $u$ of $v$ which belongs to $E_B^i$ such that $d_{E_B^i}(p,u)=d-1$. By $\boldsymbol{(P_{d-1})}$, we know that $u$ takes at most $\Pi(k,d-1)\Delta D$ actions in $e$. The $k$-boundedness of the daemon allows us to conclude that $v$ takes at most $k\times\Pi(k,d-1)\Delta D$ actions before the last action of $u$. Then, a similar reasoning to the one of the initialization part allows us to say that $v$ takes at most $\Delta D$ actions after the last action of $u$ (note that the fact that $|M(S)|\geq 2$, the construction of $D$ and the management of $dist$ variables imply that $dist_u<D-1$ after the last step of $u$). In conclusion, $v$ takes at most $k\times\Pi(k,d-1)\Delta D+\Delta D=\Pi(k,d)\Delta D$ actions in $e$, that proves $\boldsymbol{(P_d)}$.
\end{description}

As $\delta$ denotes the maximal diameter of connected components of the subsystem induced by $E_B$, then we know that $d_{E_B^i}(p,v)\leq \delta$ for any process $v$ in $E_B^i$. For any process $v$ of $E_B$, there exists $i\in\{0,\ldots,\ell\}$ such that $v\in E_B^i$. We can deduce that any process of $E_B$ takes at most $\Pi(k,\delta)\Delta D$ actions in $e$, that implies the result.
\end{proof}

\begin{lemma}\label{lem:activatedorspec}
If $\rho$ is a configuration of $\mathcal{LC}$ and $v$ is a process such that $v\in E_B$, then for any execution $e$ starting from $\rho$ either
\begin{enumerate}
\item there exists a configuration $\rho'$ of $e$ such that $spec(v)$ is always satisfied after $\rho'$, or
\item $v$ is activated in $e$.
\end{enumerate}
\end{lemma}

\begin{proof}
Let $\rho$ be a configuration of $\mathcal{LC}$ and $v$ be a process such that $v\in E_B$. By contradiction, assume that there exists an execution starting from $\rho$ such that $(i)$ $spec(v)$ is infinitely often false in $e$ and $(ii)$ $v$ is never activated in $e$.

For any configuration $\rho$, let us denote by $P_v(\rho)=(v_0=v,v_1=prnt_v,v_2=prnt_{v_1},\ldots,v_k=prnt_{v_{k-1}},p_v=prnt_{v_k})$ the maximal sequence of processes following pointers $prnt$ (maximal means here that either $prnt_{p_v}=\bot$ or $p_v$ is the first process such that there $p_v=v_i$ for some $i\in\{0,\ldots,k\}$).

Let us study the following cases:
\begin{description}
\item[Case 1:] $prnt_v\in V\setminus S_B$ in $\rho$.\\
Since $\rho\in\mathcal{LC}$, $prnt_v$ satisfies $spec(prnt_v)$ in $\rho$ and in any execution starting from $\rho$ (by Lemma \ref{lem:SBTAcontainedMax}). Hence, $prnt_v$ is never activated in $e$. If $v$ does not satisfy $spec(v)$ in $\rho$, then we have $level_v\neq met(level_{prnt_v},w_{v,prnt_v})$ or $dist_v\neq 0$ in $\rho$. Then, $v$ is continuously enabled in $e$ and we have a contradiction between assumption $(ii)$ and the strong fairness of the scheduling. This implies that $v$ satisfies $spec(v)$ in $\rho$. The fact that $prnt_v$ is never activated in $e$ and that the state of $v$ is consistent with the one of $prnt_v$ ensures us that $v$ is never enabled in any execution starting from $\rho$. Hence, $spec(v)$ remains true in any execution starting from $\rho$. This contradicts the assumption $(i)$ on $e$.
\item[Case 2:] $prnt_v\notin V\setminus S_B$ in $\rho$.\\
By the assumption $(i)$ on $e$, we can deduce that there exists infinitely many configurations $\rho'$ such that a process of $P_v(\rho')$ is enabled (since $spec(v)$ is false only when the state of a process of $P_v(\rho')$ is not consistent with the one of its parent that made it enabled). By construction, the length of $P_v(\rho')$ is finite for any configuration $\rho'$ and there exists only a finite number of processes in the system. Consequently, there exists at least one process which is infinitely often enabled in $e$. Since the scheduler is strongly fair, we can conclude that there exists at least one process which is infinitely often activated in $e$.

Let $A_e$ be the set of processes which are infinitely often activated in $e$. Note that $v\notin A_e$ by assumption $(ii)$ on $e$. Let $e'=\rho'\ldots$ be the suffix of $e$ which contains only activations of processes of $A_e$. Let $p$ be the first process of $P_v(\rho')$ which belongs to $A_e$ ($p$ exists since at least one process of $P_v$ is enabled when $spec(v)$ is false). By construction, the prefix of $P_v(\rho'')$ from $v$ to $p$ in any configuration $\rho''$ of $e$ remains the same as the one of $P_v(\rho')$. Let $p'$ be the process such that $prnt_{p'}=p$ in $e'$ ($p'$ exists since $v\neq p$ implies that the prefix of $P_v(\rho')$ from $v$ to $p$ counts at least two processes). As $p$ is infinitely often activated and as any activation of $p$ modifies the value of $level_p$ or of $dist_p$ (at least one of these two variables takes at least two different values in $e'$), we can deduce that $p'$ is infinitely often enabled in $e'$ (since the value of $level_{p'}$ is constant by construction of $e'$ and $p$). Since the scheduler is strongly fair, $p'$ is activated in a finite time in $e'$, that contradicts the construction of $p$. 
\end{description}
In the two cases, we obtain a contradiction with the construction of $e$, that proves the result.
\end{proof}

Let $\mathcal{LC^*}$ be the following set of configurations:
\[\mathcal{LC^*}=\left\{\rho\in C\left|(\rho \text{ is }S_B^*\text{-legitimate for }spec)\wedge(IM_{m_k}(\rho)=true)\right.\right\}\]

Note that, as $S_B^*\subseteq S_B$, we can deduce that $\mathcal{LC^*}\subseteq\mathcal{LC}$. Hence, properties of Lemmas \ref{lem:LCnbactions} and \ref{lem:activatedorspec} also apply to configurations of $\mathcal{LC^*}$.

\begin{lemma}\label{lem:SB*TAtimecontained}
Any configuration of $\mathcal{LC^*}$ is $(n\Pi(k,\delta)\Delta D,\Pi(k,\delta)\Delta D,S_B^*,n-1)$-TA time contained for $spec$.
\end{lemma}

\begin{proof}
Let $\rho$ be a configuration of $\mathcal{LC^*}$. As $S_B^*\subseteq S_B$, we know by Lemma \ref{lem:SBTAcontainedMax} that any process $v$ of $V\setminus S_B$ satisfies $spec(v)$ and takes no action in any execution starting from $\rho$.

Let $v$ be a process of $E_B$. By Lemmas \ref{lem:LCnbactions} and \ref{lem:activatedorspec}, we know that $v$ takes at most $\Pi(k,\delta)\Delta D$ actions in any execution starting from $\rho$. Moreover, we know that $v$ satisfies $spec(v)$ after its last action (otherwise, we obtain a contradiction between the two lemmas). Hence, any process of $E_B$ takes at most $\Pi(k,\delta)\Delta D$ actions and then, there are at most $n\Pi(k,\delta)\Delta D$ $S_B^*$-TA-disruptions in any execution starting from $\rho$ (since $|E_B|\leq n$).

By definition of a TA time contained configuration, we obtain the result.
\end{proof}

\begin{lemma}\label{lem:convergenceLC*}
Starting from any configuration, any execution of $\mathcal{SSMAX}$ reaches a configuration of $\mathcal{LC^*}$ in a finite time.
\end{lemma}

\begin{proof}
Let $\rho$ be an arbitrary configuration. We know by Lemma \ref{lem:convergenceLCMax} that any execution starting from $\rho$ reaches in a finite time a configuration $\rho'$ of $\mathcal{LC}$. 

Let $v$ be a process of $E_B$. By Lemmas \ref{lem:LCnbactions} and \ref{lem:activatedorspec}, we know that $v$ takes at most $\Pi(k,\delta)\Delta D$ actions in any execution starting from $\rho'$. Moreover, we know that $v$ satisfies $spec(v)$ after its last action (otherwise, we obtain a contradiction between the two lemmas). This implies that any execution starting from $\rho'$ reaches a configuration $\rho''$ such that any process $v$ of $E_B$ satisfies $spec(v)$. It is easy to see that $\rho''\in\mathcal{LC^*}$, that ends the proof.
\end{proof}

\begin{theorem}\label{th:possTAStrong}
$\mathcal{SSMAX}$ is a $(n\Pi(k,\delta)\Delta D,S_B^*,n-1)$-TA strongly stabilizing protocol for $spec$.
\end{theorem}

\begin{proof}
This result is a direct consequence of Lemmas \ref{lem:SB*TAtimecontained} and \ref{lem:convergenceLC*}.
\end{proof}

\section{Concluding Remarks}\label{sec:relationship}

We discuss now about the relationship between TA strong and strong stabilization on maximum metric tree construction. We characterize by a necessary and sufficient condition the set of assigned metric that allow strong stabilization. Indeed, properties on the metric itself are not sufficient to conclude on the possibility of strong stabilization: we must know information about the considered system (assignation of the metric).

Informally, it is possible to construct a maximum metric tree in a strongly stabilizing way if and only if the considered metric is strongly maximizable and if the desired containment radius is sufficiently large. More formally,

\begin{theorem}\label{th:possStrong}
Given an assigned metric $\mathcal{AM}=(M,W,mr,met,\prec,wf)$ over a system $S$, there exists a $(t,c,n-1)$-strongly stabilizing protocol for maximum metric spanning tree construction with a finite $t$ if and only if: 
\[\begin{cases}
(M,W,met,mr,\prec) \text{ is a strongly maximizable metric, and}\\
c\geq max\{0,|M(S)|-2\}
\end{cases}\]
\end{theorem}

\begin{proof}
We split this proof into two parts:

\noindent\textbf{1) Proof of the ``if'' part:}
Denote $(M,W,met,mr,\prec)$ by $\mathcal{M}$ and assume that $\mathcal{M}$ is a strongly maximizable metric and that $c\geq max\{0,|M(S)|-2\}$. We distinguish the following cases:

\begin{description}
\item[Case 1:] $|M(S)|=1$ (and hence $c\geq 0$).\\
Denote by $m$ the metric value such that $M(S)=\{m\}$. For any correct process $v$, we have $\mu(v,r)=\underset{b\in B}{min_\prec}\{\mu(v,b)\}=m$. We can deduce that it is equivalent to construct a maximum metric spanning tree for $\mathcal{M}$ and for $\mathcal{NC}$ over this system. By Theorem \ref{th:possstrongNC}, we know that there exists a $(t,0,n-1)$-strongly stabilizing protocol for this problem with a finite $t$, that proves the result.

\item[Case 2:] $|M(S)|\geq 2$ (and hence $c\geq |M(S)|-2$).\\
By Theorem \ref{th:possTAStrong}, we know that there exists a $(n\Pi(k,\delta)\Delta D,S_B^*,n-1)$-TA-strongly stabilizing protocol $\mathcal{P}$ for maximum metric spanning tree construction in this case. Denote by $\Upsilon$ the only fixed point of $\mathcal{M}$. Let $v$ be a correct process such that $v\in S_B^*$.

By definition of $S_B^*$, we have: $\mu(v,r)\prec \mu(v,b)$ for at least one Byzantine process $b$. As $\mathcal{M}$ is strictly decreasing and has only one fixed point, we can deduce that $\Upsilon\preceq\mu(v,r)$ and then $\mu(v,b)\neq\Upsilon$.

Assume that $d(v,b)>c\geq |M(S)|-2$. As $\mathcal{M}$ is strictly decreasing, has only one fixed point $\Upsilon$, and $\mathcal{M}$ has $|M(S)|$ distinct metric values over $S$, we can conclude that $\mu(v,b)=\Upsilon$. This contradiction allows us to conclude that there exists a process $b$ such that $d(v,b)\leq c$ for any correct process which belongs to $S_B^*$.

In other words, $S_B^*=\left\{v\in V|\underset{b\in B}{min}\{d(v,b)\}\leq c\right\}$ and $\mathcal{P}$ is in fact a $(n\Pi(k,\delta)\Delta D,c,n-1)$-strongly stabilizing protocol, that proves the result with $t=n\Pi(k,\delta)\Delta D$.
\end{description}

\noindent\textbf{2) Proof of the ``only if'' part:}
This result is a direct consequence of Theorem \ref{th:necessarConditionStrong} when we observe that $|M(S)|\leq |M|$ by definition.
\end{proof}

We can now summarize all results about self-stabilizing maximum metric tree construction in presence of Byzantine faults with the above table. Note that results provided in this paper fill all gaps pointed out in related works.

\footnotesize
\begin{center}
\begin{tabular}{|c||c|}
\cline{2-2}
\multicolumn{1}{c||}{}& $\mathcal{M}=(M,W,mr,met,\prec)$ is a \tabularnewline
\multicolumn{1}{c||}{}& maximizable metric \tabularnewline
\hline
\hline
$(c,f)$-strict stabilization & Impossible\tabularnewline
(for any $c$ and $f$)& (\cite{NA02c})\tabularnewline
\hline
$(t,c,f)$-strong stabilization  & Possible $\Longleftrightarrow \begin{cases}
\mathcal{M} \text{ is a strongly maximizable metric, and}\\
c\geq max\{0,|M(S)|-2\}
\end{cases}$\tabularnewline
(for $0\leq f\leq n-1$ and a finite $t$)&(Theorem \ref{th:possStrong})\tabularnewline
\hline
$(A_B,f)$-TA strict stabilization  & Impossible\tabularnewline 
(for any $f$ and $A_{B}\varsubsetneq S_B$)& (\cite{DMT10ca})\tabularnewline
\hline
$(S_B,f)$-TA strict stabilization & Possible\tabularnewline 
(for $0\leq f\leq n-1$)& (\cite{DMT10ca} and Theorem \ref{th:SSMAXstrict})\tabularnewline
\hline
$(t,A_B,f)$-TA strong stabilization & Impossible\tabularnewline 
(for any $f$ and $A_{B}\varsubsetneq S_B^*$)& (Theorem \ref{th:impTAstrong})\tabularnewline
\hline
$(t,S_B^*,f)$-TA strong stabilization & Possible\tabularnewline 
(for $0\leq f\leq n-1$ and a finite $t$) & (Theorem \ref{th:possTAStrong})\tabularnewline
\hline
\end{tabular}
\end{center}

\normalsize

To conclude about results presented in this paper, we must bring some precisions about specifications. We chose to work with a specification of the problem that consider the $dist$ variable as a O-variable. This choice may appear strong but it seems us necessary to keep the consistency of results. Indeed, impossibility results of Section \ref{sec:impossibility} can be proved with a weaker specification that does not consider the $dist$ variable as a O-variable (see \cite{DMT11ra}). On the other hand, we need the stronger specification to bound the number of disruptions of the proposed protocol. We postulate that our protocol is also TA strongly stabilizing with the weaker specification but we do no succeed to bound exactly the number of disruptions. 

The following questions are still open. Is it possible to bound the number of disruptions with the weaker specification? Is it possible to perform TA strong stabilization with a weaker daemon? Is it possible to decrease the number of disruptions without loose the optimality of the containment area?

\bibliographystyle{plain}
\bibliography{biblio}

\end{document}